\documentclass[11pt,letterpaper]{article}

\usepackage[utf8]{inputenc}
\usepackage[T1]{fontenc}
\usepackage{times}
\usepackage[english]{babel}
\usepackage[letterpaper]{geometry}
\usepackage{setspace}
\usepackage{indentfirst}

\usepackage{amsmath, amsfonts, amssymb, amsthm, mathtools, actuarialsymbol}
\theoremstyle{definition}
\newtheorem{definition}{Definition}[section]
\newtheorem{assumption}{Assumption}[section]
\newtheorem{remark}{Remark}[section]
\theoremstyle{plain}
\newtheorem{theorem}{Theorem}[section]
\newtheorem{lemma}[theorem]{Lemma}
\newtheorem{corollary}[theorem]{Corollary}
\numberwithin{equation}{section}
\newtheorem{proposition}[theorem]{Proposition}

\usepackage{graphicx}
\usepackage{booktabs}
\usepackage{caption}
\usepackage{subcaption}
\usepackage[linesnumbered,ruled,vlined]{algorithm2e}
\usepackage{longtable}

\usepackage{authblk}
\usepackage{fancyhdr}
\usepackage[authoryear]{natbib}
\usepackage{hyperref}

\usepackage[capitalise]{cleveref}

\hypersetup{
    colorlinks=true,
    allcolors=blue,
    linktoc=all,
    pdfborder={0 0 0},
}

\newcommand{\instdiscount}{\eta_t}

\title{Habit Formation, Labor Supply, and the Dynamics of Retirement and Annuitization}

\author[1]{Criscent Birungi}
\author[1]{Cody Hyndman\thanks{Corresponding author. Emails:
     \href{mailto:criscent.birungi@concordia.ca}{criscent.birungi@concordia.ca},
     \href{mailto:cody.hyndman@concordia.ca}{cody.hyndman@concordia.ca}}}

\affil[1]{Department of Mathematics and Statistics \\ Concordia University \\ Montr\'eal, QC, Canada}

\date{February 1, 2026}

\begin{document}
\maketitle

\begin{abstract}
  \noindent
  The decision to annuitize wealth in retirement planning has become increasingly complex due to rising longevity risk and changing retirement patterns, including increased labor force participation at older ages. While an extensive literature studies consumption, labor, and annuitization decisions, these elements are typically examined in isolation.  This paper develops a unified stochastic control and optimal stopping framework in which habit formation and endogenous labor supply shape retirement and annuitization decisions under age-dependent mortality.  We derive optimal consumption, labor, portfolio, and annuitization policies in a continuous-time lifecycle model. The solution is characterized via dynamic programming and a Hamilton–Jacobi–Bellman variational inequality.  Our results reveal a rich sequence of retirement dynamics.  When wealth is low relative to habit, labor is supplied defensively to protect consumption standards.  As wealth increases, agents enter a work-to-retire phase in which labor is supplied at its maximum level to accelerate access to retirement. Human capital acts as a stabilizing asset, justifying a more aggressive pre-retirement investment portfolio, followed by abrupt de-risking upon annuitization. Subjective mortality beliefs are a key determinant in shaping retirement dynamics. Agents with pessimistic longevity beliefs rationally perceive annuities as unattractive, leading them to avoid or delay annuitization. This framework provides a behavior-based explanation for low annuity demand and offers guidance for retirement planning jointly linking labor supply, portfolio choice, and the timing of annuitization. 
\end{abstract}

\vspace{1em}

\vspace{1em}

\section{Introduction}
\noindent An annuity is a contract that provides a buyer with a guaranteed, regular income. This conversion of an investment into a steady stream of payments is known as annuitization. As \cite{buttarazzi2025optimal} notes, this is a major retirement decision where individuals trade the potential for investment growth for the long-term financial stability of a guaranteed income for life. This decision is increasingly complex due to rising longevity risk, a demographic shift evidenced by growing labor force participation among older age groups as noted by \cite{gao2022optimal}. For instance, the Bureau of Labor Statistics projects that by 2033, the labor participation rate will exceed $10\%$ for those aged 75 and older and surpass $30\%$ for those aged 65. The presence of flexible labor income, which can be used to bridge employment gaps or supplement retirement welfare \citep[see][]{klotz2021, lorenz2021}, transforms conventional retirement planning.

The primary challenge is to determine how an agent can maximize lifetime utility from consumption and leisure, given deterministic age-dependent force of mortality, while managing wealth optimally. Formally, the problem is a coupled continuous-time stochastic control and optimal stopping problem, in which consumption, labor supply, portfolio choice, and the timing of irreversible annuitization are jointly determined. This study addresses this optimization problem within the framework of stochastic optimal control. The agent's wealth, supplemented by wage income, is governed by a stochastic differential equation (SDE) influenced by their strategic choices in investment ($\pi_t$), consumption ($c_t$), and labor ($b_t$). The option to irrevocably annuitize wealth introduces an optimal stopping component that is jointly determined with consumption, labor, and portfolio decisions.

We employ the dynamic programming approach, specifically formulating the problem as a Hamilton-Jacobi-Bellman Variational Inequality (HJBVI), as the agent must simultaneously solve for the optimal policies and the optimal time to stop (i.e., to annuitize). This HJB approach, developed in the seminal works of \cite{merton1969lifetime, merton1971optimum_na} and extended to stopping problems by \cite{shreve1998stochastic}, provides a natural framework for jointly characterizing optimal policies and annuitization timing.  It complements alternative frameworks such as the duality and martingale methods, which have been foundational to modern portfolio and consumption/investment theory \citep{rockafellar1998variational, karatzas2000utility, gao2022optimal}.

Our work builds on two rich streams of literature. The first explores optimal consumption, leisure, and investment choices \citep[see][]{cvitanic1992convex, labbe2007convex, choi2008optimal, barucci2012optimal, koo2013optimal, lee2015optimal, peng2023optimal}. A significant portion of this literature, including \cite{gerrard2012choosing} and \cite{gao2022optimal}, simplifies the problem by assuming a constant mortality rate. This assumption, while tractable, abstracts from the strongly age-dependent nature of mortality risk.  A second, parallel stream of literature incorporates the psychological realism of \textit{habit formation}, where utility is derived not from the level of consumption, but from its level relative to an accustomed habit \citep{sundaresan1989intertemporally, constantinides1990habit, detemple1991asset, Dybvig1995Duesenberry}. This feature introduces a powerful consumption-smoothing motive and a "defensive" demand for assets. However, these models often do not include the complexities of flexible labor supply or the optimal annuitization decision.

This paper bridges these gaps by characterizing a unified stochastic control and optimal stopping framework in which habit formation and endogenous labor supply generate state-dependent retirement and annuitization regimes.  We derive the HJBVI for this problem and, by employing a dimensionality reduction technique, derive semi-analytical solutions for the value function and the corresponding optimal policies for consumption, labor supply, and investment. We then provide analysis with rigorous characterization of the interplay between consumption habits, mortality risk, labor-leisure trade-offs, and investment behavior under the option to annuitize.

The rest of this paper is structured as follows. Section \ref{market_model} discusses the market model, including the economic background and mathematical formulations. Section \ref{sec:DPP} discusses the HJBVI and the proposed methods for determining optimality and provides rigorous proofs and characterizes semi-analytical solutions for the value function and the optimal policies. Section \ref{Numerical_Implementation} presents numerical implementation results and key findings. Finally, Section \ref{Conclusion_Recommendations} presents conclusions and recommendations for future research.

\section{Economic background}\label{market_model}
\noindent This section presents the financial market, the key processes governing the agent's decisions, and the formulation of the optimization problem.
\subsection{The Financial Market}
\noindent We consider a financial market on a filtered probability space $(\Omega, \mathcal{F}, \{\mathcal{F}_t\}_{t \geq 0}, P)$ satisfying the usual conditions and supporting a one-dimensional standard Brownian motion $\{W_t\}_{t \geq 0}$. The market contains two continuously traded assets: a risk-free money market account, $S^0$, and a risky stock, $S^1$. Their prices evolve as
\begin{align}
dS_t^0 &= r S_t^0 \, dt, \qquad S_0^0 = 1, \\
dS_t^1 &= \mu S_t^1 \, dt + \sigma S_t^1 \, dW_t, \qquad S_0^1 > 0,
\end{align}
where $r>0$ is the constant risk-free rate, and $\mu \in \mathbb{R}$ and $\sigma>0$ are the constant expected return and volatility of the risky asset, respectively. Let $\phi_t^0$ and $\phi_t^1$ denote the number of units held at time $t$ of the risk-free and risky assets, respectively. The total wealth is
\begin{equation}\label{total_wealth_labour}
X_t = \phi_t^0 S_t^0 + \phi_t^1 S_t^1.
\end{equation}
In this framework, the agent makes decisions on consumption $c_t \geq 0$, investment, and labor supply $b_t \geq 0$. The labor supply generates income at a constant wage rate $w>0$. We assume the agent holds a \emph{self-financing} portfolio.
The self-financing condition is modified to include both consumption withdrawals and labor income infusions as
\begin{equation}\label{self_financing_labour}
dX_t = \phi_t^0 \, dS_t^0 + \phi_t^1 \, dS_t^1 - c_t \, dt + w b_t Z_t \, dt
\end{equation}
where $Z_{t}$ is the consumption habit.  The condition in \eqref{self_financing_labour} implies that changes in wealth arise only from capital gains, consumption, and labor income, with no other external infusion or withdrawal of funds. To derive the wealth dynamics, we define the \emph{dollar amount} invested in the risky asset as $\pi_t := \phi_t^1 S_t^1$. From the definition of total wealth \eqref{total_wealth_labour}, the amount in the risk-free asset is $\phi_t^0 S_t^0 = X_t - \pi_t$. Substituting the asset dynamics and these portfolio definitions into the self-financing condition \eqref{self_financing_labour}, we obtain
\begin{equation} \label{wealth_sde_labour}
dX_t = \left[ rX_t + \pi_t(\mu - r) - c_t + w b_t Z_t \right] dt + \sigma \pi_t \, dW_t.
\end{equation}
For the remainder of the paper, we define the market price of risk as $\theta:=(\mu-r)/\sigma$. Equation \eqref{wealth_sde_labour} describes the evolution of the agent’s wealth over time, incorporating investment decisions $(\pi_t)$, consumption plans $(c_t)$, and labor income $(wb_t Z_t)$.

\subsection{Problem Formulation}
\noindent
We begin by establishing the foundational assumptions for the optimization problem.

\begin{assumption}[] \label{ass:setup}
We assume the following initial conditions and market structure
\begin{enumerate}%
 \item \textit{Market Parameters:} The financial market is complete and free of arbitrage. The parameters $r, \mu, \sigma, w$ are all positive constants.
 \item \textit{Initial Conditions:} The retiree starts with an initial wealth $X_0 = x > 0$ and an initial habit level $Z_0 = z > 0$.
\end{enumerate}
\end{assumption}

\subsubsection{State Processes: Wealth and Habit}
\noindent The agent's financial journey is captured by two state variables: their \textit{wealth ($X_t$)} and their \textit{consumption habit ($Z_t$)}. The evolution of the agent's wealth over time is governed by the SDE presented in \eqref{wealth_sde_labour}. The consumption habit, $Z_t$, models how the agent's past spending influences their current consumption desires. Following the framework in \cite{angoshtari2023optimal}, the dynamics of this habit formation process are described by the ordinary differential equation
\begin{equation}\label{enq:Habit_model}
    dZ_t = \tilde{\rho} (c_t - Z_t)dt,
\end{equation}
where $c_t$ represents the agent's consumption rate, while the parameter $\tilde{\rho} > 0$ determines how quickly their habit adapts to changes in their spending. A central feature of this model is a floor on consumption, which prevents the agent from drastically reducing their spending. This constraint ensures that their consumption rate, $c_t$, must always be greater than or equal to a specified proportion, $\alpha$, of their habit level, $Z_t$
\begin{equation} \label{eq:consumption_floor}
    \alpha Z_t \leq c_t, \quad \text{for a constant } \alpha \in (0,1].
\end{equation}
Here, the parameter $\alpha$ quantifies the \textit{addictiveness} of the habit. As $\alpha$ approaches 1, the habit is considered completely addictive, compelling the agent to consume at a rate very close to their established habit. Conversely, a value of $\alpha$ near 0 corresponds to a non-addictive habit, granting the agent greater flexibility in their consumption choices.

\begin{assumption}[Financial Viability] \label{ass:viability}
The model parameters must satisfy the inequality $r > \tilde{\rho}(1-\alpha)$. This assumption is necessary to guarantee that the minimum required consumption plan for the agent, where $c_t = \alpha Z_t$ for all $t$, is affordable.
\end{assumption}

\subsubsection{Deterministic force of mortality}\label{Deterministic_Force_of_Mortality}

\noindent Let $n$ denote the agent's current age, and let $T_n$ represent the agent's uncertain remaining lifetime. In \cite{gao2022optimal} and \cite{gerrard2012choosing}, it is assumed that $T_n$ follows a distribution with a constant force of mortality $\delta > 0$. Under this assumption, the probability of surviving $t$ years is given by
\begin{equation}\label{constant_mortality}
P(T_n > t) = \exp\left(-\delta t\right).
\end{equation}
However, a more realistic model would incorporate an age-dependent force of mortality, particularly one that increases with age. Following the argument in \cite{ashraf2023voluntary}, we introduce an age-dependent force of mortality and assume the deterministic Gompertz  mortality law. This law assumes the rate of change in the force of mortality, $\delta_t$, is proportional to its current level, leading to exponential growth over time. Mathematically, $\delta_t$ follows the ordinary differential equation (ODE)
\begin{equation}\label{eq:age_dependent_mortality_ODE}
\mathrm{d}\delta_t = \tilde{\theta} \delta_t \,\mathrm{d}t.
\end{equation}
The solution to the ODE in \eqref{eq:age_dependent_mortality_ODE} yields the time-dependent force of mortality: $\delta_t = \delta_0 e^{ \tilde{\theta} t}$. This functional form is consistent with the usual actuarial representation at age $n_t$ and time $t$
\begin{equation}\label{eq:Gompertz_mortality}
\delta_t = \frac{1}{\lambda} \exp\left(\frac{n_t - m}{\lambda}\right),
\end{equation}
where $n_t$ is the age at time $t$, $m$ is the modal value of life (in years), and $\lambda$ is the dispersion parameter (in years). The term $\tilde{\theta}$ is referred to as the Gompertz aging rate and coincides with the inverse of the dispersion parameter (i.e., $\tilde{\theta} = 1/\lambda$). In this age-dependent case, the deterministic probability of surviving $t$ years is expressed as
\begin{equation}\label{eq:age_dependent_mortality}
\px[t]{n}[(\delta)]=P(T_n > t)=\exp\left(-\int_{0}^{t} \delta_y \, dy\right),
\end{equation}
where $\delta_y$ represents the force of mortality at time $y$ (age $n+y$). The deterministic age-dependent force of mortality is independent of the fund's level of evolution, simplifying the subsequent analysis.

\vspace{1em}

\noindent The survival probability from time $t$ to $s$ (where $s > t$) for an individual currently aged $n$, denoted $p(t; s; \delta_t) = P(T_n > s-t \mid T_n > t)$, exhibits a key simplifying property under a deterministic force of mortality
\begin{align}
p(t; s; \delta_t)
&= \exp\!\left(-\int_t^s \delta_y\,dy\right)
= \frac{\exp\!\left(-\int_0^s \delta_y\,dy\right)}{\exp\!\left(-\int_0^t \delta_y\,dy\right)}
= \frac{p(s;\delta_0)}{p(t;\delta_0)} .
\label{eq:deterministic_survival}
\end{align}
The identity in equation \eqref{eq:deterministic_survival} generally fails once the mortality intensity becomes stochastic. Nonetheless, since calibration is typically carried out at time $0$, that is, only to $p(t;\delta_0)$, this discrepancy does not pose practical difficulties.

\subsubsection{Admissible Strategies and the Optimization Problem}
\noindent To ensure the optimization problem is well-posed, we must formally define the set of admissible control strategies. The agent selects a portfolio process ($\pi_t$), a consumption rate ($c_t$), and a labor supply rate ($b_t$), which are all assumed to be $\mathcal{F}_t$-progressively measurable processes. Following the framework in \citet{lee2015optimal}, we define the sets of allowable control processes.\\

\noindent The set of admissible portfolio strategies $\Pi$ is defined as
\begin{equation}
\Pi := \left\{ \pi : [0, \infty) \times \Omega \to \mathbb{R} \,\middle|\, \pi \text{ is } \mathcal{F}_t\text{-prog. measurable and } \int_0^T \|\pi_t\|^2 \, dt < \infty, \text{ a.s. } \forall T > 0 \right\}.
\end{equation}

\noindent The set of admissible consumption processes $\mathcal{C}$ is defined as
\begin{equation}
\mathcal{C} := \left\{ c : [0, \infty) \times \Omega \to \mathbb{R} \;\middle|\; c \text{ is } \mathcal{F}_t\text{-prog. measurable and } \int_0^T c_t \, dt < \infty \text{ a.s. } \forall T > 0 \right\}.
\end{equation}

\noindent The set of admissible labor supply processes $\mathcal{B}$ is defined as
\begin{equation}
\mathcal{B} := \left\{ b : [0, \infty) \times \Omega \to [0, \bar{b}] \;\middle|\; b \text{ is } \mathcal{F}_t\text{-prog. measurable and } \int_0^T b_t \, dt < \infty \text{ a.s. } \forall T > 0 \right\}.
\end{equation}
Let $\mathcal{T}$ be the set of all $\mathcal{F}_t$-stopping times. An admissible strategy is a tuple $(\pi, c, b, \tau) \in \Pi \times \mathcal{C} \times \mathcal{B} \times \mathcal{T}$ that satisfies several key conditions.

\begin{assumption}[Admissibility]\label{ass:admissibility_conditions}
A control strategy $(\pi, c, b, \tau)$ is considered admissible for an initial state $(x,z)$ if it satisfies the following conditions
\begin{enumerate}
    \item \textit{Regularity:} The control processes must belong to their respective admissible sets, i.e., $(\pi, c, b) \in \Pi \times \mathcal{C} \times \mathcal{B}$.
    \item \textit{Habit Constraint:} The consumption rate must remain above a minimum proportion $\alpha$ of the habit level: $c_t \ge \alpha Z_t$ for almost every $(t, \omega) \in [0, \tau]$.
    \item \textit{Solvency Constraint:} The agent's wealth process $X_t^{(\pi,c,b)}$ must remain non-negative throughout the investment horizon: $X_t^{(\pi,c,b)} \ge 0$ for all $t \in [0, \tau]$ almost surely. This prevents bankruptcy.
    \item \textit{Stopping Time:} The time $\tau$ must be an $\{\mathcal{F}_t\}$-stopping time. This ensures that the decision to annuitize is non-anticipatory and based only on information available up to that time.
\end{enumerate}
\end{assumption}

\noindent The set of all such admissible strategies for a given initial state $(t, x, z)$ is denoted by $\mathcal{A}(t, x, z)$. 
\subsubsection{Leisure process}

\noindent The consumer maximizes utility derived from consumption $c_t$ and leisure $l_t$, subject to choices over consumption $c_t$, portfolio allocation $\pi_t$, and labor supply $b_t$. Retirement occurs at the optimally chosen stopping time $\tau$.  

A consumer is endowed with two levels of leisure. We adopt a two-level leisure specification in which the agent enjoys a lower level of leisure before retirement and a higher level afterward, consistent with \cite{huang2012yaari} and \cite{ashraf2023voluntary}. The leisure process is defined as

\begin{equation}\label{eq:leisure_constraint}
l_t =
\begin{cases}
l_1, & \text{for } 0 \le t < \tau \quad \text{(pre-retirement)}, \\
\bar{l}, & \text{for } t \ge \tau \quad \text{(post-retirement)},
\end{cases}
\end{equation}

\noindent where $\bar{l}> l_1 \ge 1$ (following \cite{choi2008optimal} and \cite{koo2013optimal}). Because labor is supplied only before retirement ($t < \tau$), we set pre-retirement leisure to $l_1 = 1$ and take post-retirement leisure to satisfy $\bar{l} > 1$. It is worth emphasizing that, in this setup, leisure is not expressed in units of time \cite{ashraf2023voluntary}; rather, the parameter $l_t$ reflects the agent's relative utility or enjoyment of consumption after retirement compared with the working phase.\\

\noindent The instantaneous utility function is multiplicatively separable and given by
\begin{equation}\label{eqn:utility_fun_1}
 U(c_t, b_t) = \frac{1}{1-\gamma} \left( c_t (\bar{l} - b_t)^{\psi} \right)^{1-\gamma},
\end{equation}
where $\gamma > 0, \gamma \neq 1$ is the coefficient of relative risk aversion over the composite consumption-leisure good, $\bar{l} - b_t$ is leisure, and $\psi > 0$ is the \textit{preference weight for leisure}.

The agent's objective is to maximize the total expected discounted utility. The objective function for any strategy $(\pi,c,b,\tau) \in \mathcal{A}(t,x,z)$ is

\begin{equation}\label{eq:objective_functional_habit}
J(t,x,z;\pi,c,b,\tau) = \mathbb{E}_{t,x,z}\left[ \int_t^{\tau} e^{-\int_t^s (\beta + \delta_u) du} U(c_s, b_s) ds + e^{-\int_t^{\tau} (\beta + \delta_u) du} G(X_{\tau}) \right],
\end{equation}
where $\beta > 0$ is the subjective time preference rate and $G(X_\tau)$ is the post-annuitization value function. For the expectation in equation \eqref{eq:objective_functional_habit} to be well-defined, we require that

\begin{equation}
\mathbb{E}_{t,x,z}\left[ \int_t^\tau e^{-\int_t^s (\beta + \delta_u) du} U^-(c_s, b_s) \, ds + e^{-\int_t^{\tau} (\beta + \delta_u) du} G^-(X_{\tau}) \right] < \infty,
\end{equation}
where $f^- = \max(-f, 0)$ denotes the negative part of a function $f$. The optimization problem is to find the value function $V(t, x, z)$
\begin{equation}\label{eq:value_function_def}
V(t,x,z) = \sup_{(\pi,c,b,\tau) \in \mathcal{A}(t,x,z)} J(t,x,z;\pi,c,b,\tau).
\end{equation}
The expression in equation \eqref{eq:value_function_def} is a stochastic optimal control problem with state variables $(X_t, Z_t)$ and control variables $(\pi_t, c_t, b_t, \tau)$. The solution involves finding the optimal policies $(\pi_t^*, c_t^*, b_t^*, \tau^*)$ that attain the supremum. The remainder of the paper is dedicated to solving this problem.

\begin{remark}
We assume that $V(t,x,z) < \infty$ for all $(t,x,z)$. As noted in \cite{bertsekas1996stochastic}, a sufficient condition for this is that the utility functions in equation \eqref{eq:objective_functional_habit} satisfies a growth condition. For our problem, this would mean that
\begin{equation}\label{sufficient_condition_V}
\max\{U(c,b), G(x)\} \leq K_1 + K_2 (x+z)^{\hat{\eta}}
\end{equation}
for some constants $K_1 > 0, K_2 > 0$ and $\hat{\eta} \in (0,1)$.
\end{remark}

\section{Dynamic Programming and the HJB Variational Inequality}\label{sec:DPP}
\subsection{Dynamic Programming Principle}\label{sec:DDP_}
\noindent The \textit{Dynamic Programming Principle (DPP)} is used to solve the problem in equation \eqref{eq:value_function_def}. For a combined optimal control and stopping problem, the DPP asserts that at any point in time, the optimal strategy must be the better of two choices: either to stop immediately (i.e., annuitize) or to continue for an infinitesimal time period and then proceed optimally from the new state. The optimal path is thus composed of a series of such optimal decisions.

\begin{proposition}[Dynamic Programming Principle]\label{thm:Dynamic_programming_principle_habit}
The value function $V(t, x, z)$ satisfies the following dynamic programming equation for any small time interval $\Delta t > 0$
\begin{equation} \label{eq:dpp}
\begin{split}
    V(t, x, z) = \max \Biggl\{ G(x), \sup_{(\pi, c, b)} \mathbb{E}_{t,x,z} \biggl[ &\int_t^{t+\Delta t} e^{-\int_t^s (\beta + \delta_u) du} U(c_s,b_s) ds \\
    &+ e^{-\int_t^{t+\Delta t} (\beta + \delta_u) du} V(t+\Delta t, X_{t+\Delta t}, Z_{t+\Delta t}) \biggr] \Biggr\}.
\end{split}
\end{equation}
\end{proposition}
\begin{proof}
  See Appendix~\ref{app:DPP_Proof}.
\end{proof}

\subsection{The Hamilton-Jacobi-Bellman Variational Inequality}
\noindent Since equation \eqref{eq:value_function_def} is a combined optimal control and stopping problem, the value function must satisfy the Hamilton-Jacobi-Bellman Variational Inequality (HJBVI).  Let $\instdiscount = \beta + \delta_t$ be the instantaneous effective discount rate. The value function $V(t,x,z)$ must satisfy the HJBVI
\begin{equation}\label{eq:hjb_vi_labor_corrected}
\max \left( \sup_{\pi, c \ge \alpha z, b} \left\{ U(c, b) - \instdiscount V + \mathcal{L}^{\pi, c, b}V \right\}, \quad G(x) - V(t, x, z) \right) = 0.
\end{equation}
where the infinitesimal generator $\mathcal{L}^{\pi,c,b}$ for the state process $(X_t, Z_t)$ is given by
\begin{equation}\label{eq:generator_full_labor}
\mathcal{L}^{\pi, c,b} V = \frac{\partial V}{\partial t} + \left[ rX + \pi(\mu - r) - c + w b Z \right] \frac{\partial V}{\partial x} + \tilde{\rho}(c-Z) \frac{\partial V}{\partial z} + \frac{1}{2}\sigma^2 \pi^2 \frac{\partial^2 V}{\partial x^2}.
\end{equation}

\subsection{Dimensionality Reduction}\label{sec:Dimensionality_Reduction}
\noindent Solving the HJBVI in equation \eqref{eq:hjb_vi_labor_corrected} as a free-boundary problem with two state variables can be challenging. We employ \textit{dimensionality reduction}, leveraging the homothetic properties of the utility function. This allows us to reframe the stochastic optimal control problem in terms of a single state variable: the wealth-to-habit ratio, $y=x/z$. We define this ratio as
\begin{equation}\label{eq:y_t}
    y_t := \frac{X_t}{Z_t}.
\end{equation}
This reduction yields a one-dimensional free-boundary problem that admits semi-analytical solutions, with optimal policies characterized explicitly up to a finite set of constants determined by smooth-pasting conditions. The dynamics of $y_{t}$, which are central to the reformulation of the problem  are obtained by It\^{o}'s formula applied to the quotient of $X_t$ and $Z_t$.

\begin{lemma}[SDE for the Wealth-to-Habit Ratio]\label{SDE_Wealth-to-Habit_Ratio}
Let the agent's wealth process $X_t$ and habit process $Z_t$ be governed by the SDEs in equation \eqref{wealth_sde_labour} and \eqref{enq:Habit_model}, respectively. Define the \textit{portfolio weight} $p_t := \pi_t / X_t$, and the \textit{consumption-to-habit ratio} $\kappa_t := c_t / Z_t$. Then, the SDE governing the wealth-to-habit ratio, $y_t := X_t / Z_t$, is given by
\begin{equation} \label{eq:y_dynamics_final_labor_corrected}
    \mathrm{d}y_t = \Big[(r+\tilde{\rho})y_t + p_t y_t(\mu-r) - \kappa_t(1+\tilde{\rho} y_t) + w b_t \Big] \mathrm{d}t + \sigma p_t y_t\,\mathrm{d}W_t,
\end{equation}
\end{lemma}

\begin{proof}
To derive the SDE for $y_t$, we apply Itô's lemma for a quotient. Since the habit process $Z_t$ has dynamics governed only by a $\mathrm{d}t$ term, it is a process of finite variation, so the Itô quotient rule simplifies to
\[
\mathrm{d}\left(\frac{X_t}{Z_t}\right) = \frac{1}{Z_t}\mathrm{d}X_t - \frac{X_t}{Z_t^2}\mathrm{d}Z_t.
\]
Substituting the SDEs for $\mathrm{d}X_t$ \eqref{wealth_sde_labour} and $\mathrm{d}Z_t$ \eqref{enq:Habit_model}, we have
\begin{align}
\mathrm{d}y_t &= \frac{1}{Z_t} \left[ (rX_t + \pi_t(\mu-r) - c_t + w b_t Z_t)\,\mathrm{d}t + \sigma\pi_t \,\mathrm{d}W_t \right] 
- \frac{X_t}{Z_t^2}\left[ \tilde{\rho}(c_t-Z_t)\,\mathrm{d}t \right] \\
&= \left[ r\frac{X_t}{Z_t} + \frac{\pi_t}{Z_t}(\mu-r) - \frac{c_t}{Z_t} + w b_t - \tilde{\rho}\frac{X_t}{Z_t}\left(\frac{c_t}{Z_t}-1\right)\right]\mathrm{d}t 
+ \sigma\frac{\pi_t}{Z_t} \,\mathrm{d}W_t. \label{dy_t}
\end{align}
We now use the dimensionless ratios $y_t = X_t/Z_t$, $p_t = \pi_t/X_t$, and $\kappa_t=c_t/Z_t$. This implies $\pi_t/Z_t = p_t X_t / Z_t = p_t y_t$. Substituting these definitions into equation \eqref{dy_t} gives
\begin{align*}
\mathrm{d}y_t &= \left[ ry_t + p_t y_t(\mu-r) - \kappa_t + w b_t - \tilde{\rho} y_t(\kappa_t-1) \right]\mathrm{d}t + \sigma p_t y_t \mathrm{d}W_t \\
&= \left[ ry_t + p_t y_t(\mu-r) - \kappa_t - \tilde{\rho} y_t \kappa_t + \tilde{\rho} y_t + w b_t \right]\mathrm{d}t + \sigma p_t y_t \mathrm{d}W_t \\
&= \left[(r+\tilde{\rho})y_t + p_t y_t(\mu-r) - \kappa_t(1+\tilde{\rho} y_t) + w b_t \right] \mathrm{d}t + \sigma p_t y_t\,\mathrm{d}W_t.
\end{align*}
This completes the derivation.
\end{proof}

\noindent The dynamics for the habit process $Z_t$ can also be rewritten in terms of $\kappa_t$
\begin{equation}\label{eq:habit_dynamics_kappa}
    dZ_t = \tilde{\rho}(c_t - Z_t)dt = \tilde{\rho} Z_t (\kappa_t - 1)dt.
\end{equation}

\noindent The utility function from \eqref{eqn:utility_fun_1} can be written in terms of the consumption-to-habit ratio $\kappa_t$ as
\begin{equation}\label{eqn:utility_fun_2}
    U(c_t, b_t) = \frac{(Z_t)^{1-\gamma}}{1-\gamma} \left( \kappa_t (\bar{l} - b_t)^{\psi} \right)^{1-\gamma}.
\end{equation}

\noindent We now formally define the agent's utility-maximization problem, as in \eqref{eq:objective_functional_habit}. %
Following the framework of \cite{gerrard2012choosing}, the agent controls the consumption rate, the portfolio allocation, and the timing of annuitization.  A fund of size $x$ is converted into an annuity of $kx$, where the annuity factor $k>r$, and this decision is irreversible. If the fund is exhausted before this point, all economic activity, including further investment or withdrawal, ceases.

\begin{remark}[Annuitization Value Function]
For the annuitization phase, we assume the value function at annuitization
is analogous to the value function from the classical Merton problem following \cite{koo2013optimal} framework. The terminal value in \eqref{eq:objective_functional_habit} is 
\begin{equation}
J(y,Z) = \frac{e^{-\rho \tau}}{\rho_{\tau}} \frac{(k y Z)^{1-\gamma}}{1-\gamma} 
\end{equation}

\noindent By applying the homothetic transformation $G(y) = \frac{1-\gamma}{Z^{1-\gamma}}J(y,Z)$ and substituting $X=yZ$, we derive the terminal condition for the transformed problem, 
\begin{equation}\label{eq:post_annuitization_value_1}
G(y) = \frac{1-\gamma}{Z^{1-\gamma}} \left( \frac{e^{-\rho \tau}}{\rho_{\tau}} \frac{(k y Z)^{1-\gamma}}{1-\gamma} \right) = e^{-\rho \tau} \frac{(ky)^{1-\gamma}}{\rho_{\tau}}.
\end{equation}

\end{remark}

\noindent The agent's problem is to choose the control processes $(p_t, \kappa_t, c_t)_{t \ge 0}$ to maximize the objective function. By defining an effective discount rate that incorporates both the subjective time preference \(\beta\) and the age-dependent force of mortality \(\delta_t\), the value function in \eqref{eq:objective_functional_habit} can be expressed as
\begin{equation}\label{eq:objective_functional_final}
V(y)= \sup_{(p_s, \kappa_s, b_s)_{s \ge t}} \mathbb{E} \left[\int_0^\tau e^{-(\rho_s) s} U(b_s, c_s) ds  + G(y) \right].
\end{equation} 
where \(\rho_s\) is the effective cumulative discount rate, given by
\begin{equation}\label{eq:age-dependent_force_of_mortality}
\rho_s = \int_0^s (\beta + \delta_u) \, du.
\end{equation}
and \(U(c_s, b_s)\) is the instantaneous utility function as defined in \eqref{eqn:utility_fun_2} and \(G(y)\) is the value derived from the wealth annuitized at time \(\tau\) as defined in \eqref{eq:post_annuitization_value_1}. The remainder of the paper, we focus on equation \eqref{eq:objective_functional_final}.

\begin{definition}\label{defn:optimal_retirement_wealth_threshold}
The \textit{optimal retirement wealth threshold}, denoted by $y^*$, is the critical level of wealth at which an agent chooses to retire to maximize their lifetime utility.
\end{definition}
\begin{definition}\label{defn:subsistence_consumption_wealth_threshold}
The \textit{labor constraint wealth threshold}, denoted by $\tilde{y}$, is the level of wealth required to ensure the optimal labor supply $b^*$ is strictly less than the upper limit $\bar{b}$. For $\tilde{y} \leq y \leq y^*$, the labor constraint is binding, $b^*=\bar{b}$.
\end{definition}

\noindent Therefore, we’ll break this optimization problem into two interlinked problems: Continuation Region ($0<y < y^*$) and Stopping Region ($y \ge y^*$). We will solve the optimization problem in \eqref{eq:objective_functional_final} by separating it into two interconnected phases, distinguished by an optimal retirement wealth threshold, $y^*$. The decision to transition between phases occurs at the optimal stopping time, defined as
\begin{equation}\label{eq:optimal_stopping_time}
\tau_{y}^{*} = \inf\{t \ge 0 : y_{t} \ge y^{*}\}.
\end{equation}
This framework establishes two distinct operational regions
\begin{itemize}
    \item \textit{Continuation Region ($0 < y < y^*$):} In this phase, the agent actively manages their consumption and investment portfolio while working. In this region, the HJB equation should hold as an equality.
    
    \item \textit{Stopping Region ($y \ge y^*$):} Once the state variable reaches the threshold $y^*$, the process is terminated, and the agent annuitizes their wealth. This is also called full retirement period. Here, the value function $V(y)$ is equal to the terminal payoff.
\end{itemize}

\begin{remark}[Stationary condition of the  value function]\label{eq:Stationary condition of the  value function}
Let $V(t, y)$ be the value function in \eqref{eq:objective_functional_final}, this problem is an optimal stopping problem, and the value function must satisfy the HJB variational inequality. We solve this problem by treating the agent's age $t$ as a fixed parameter. This makes the instantaneous force of mortality $\delta_t$ and the effective discount rate $\instdiscount = \beta + \delta_t$ constants for a given HJB equation. The resulting value function $V(y; t)$ is thus 'stationary' with respect to the $y$ variable. For notational simplicity, we write $V(y)$ and $\instdiscount$, where $t$ is implicit.

\begin{equation} \label{eq:hjb_variational_y}
\max \left\{ \sup_{p, \kappa, b} \left[ u_1(\kappa, b) + \mathcal{L}^{p, \kappa, b} V(y) \right] - \instdiscount\, V(y), \quad G(y) - V(y) \right\} = 0,
\end{equation}
where the generator $\mathcal{L}^{p, \kappa, b} V(y)$ is defined as
\[
\mathcal{L}V(y) = V'(y) \left[ (r+\rho)y + p y(\mu-r) - \kappa(1+\rho y) + w b \right] + \frac{1}{2} V''(y) \sigma^2 p^2 y^2,
\]
and $G(y)$ is the value derived from annuitizing at wealth-to-habit ratio $y$.
\end{remark}

\begin{remark}[Structure of the optimal strategy]\label{eq:Structure_of_the_optimal_strategy}
The HJB variational inequality in \eqref{eq:hjb_variational_y} partitions the state space ($y \ge 0$) into two distinct regions, separated by the optimal retirement wealth threshold $y^*$.
\begin{enumerate}
    \item \textit{Continuation Region ($y < y^*$):} For wealth-to-habit ratios below the threshold, it is optimal to continue working. Here, $V(y) > G(y)$ and the value function solves the HJB equation
    \[  
    \instdiscount V(y) = \sup_{p, \kappa, b} \left[ u_1(\kappa, b) + \mathcal{L}^{p, \kappa, b} V(y) \right].
    \]
    \item \textit{Stopping (Annuitization) Region ($y \ge y^*$):} Once the ratio reaches or exceeds $y^*$, it is optimal to annuitize. The value function is equal to the annuitization value
   $ V(y) = G(y)$.
\end{enumerate}
\end{remark}

\begin{theorem}[Value Function]\label{thm:value_function}
Given that the wealth-to-habit ratio is $y_t = y$, assume the necessary regularity conditions in \ref{ass:admissibility_conditions} hold. The value function $V(y)$ is given by a piecewise function
\begin{equation}\label{eq:value_function_y}
V(y) =
\begin{cases}
V_{\mathrm{int}}(y), & \text{if } y < \tilde{y}, \\
V_{\bar{b}}(y), & \text{if } \tilde{y} \leq y < y^*, \\
G(y), & \text{if } y \geq y^*,
\end{cases}
\end{equation}
\noindent where $G(y)$ is the value function in the annuitization region, given by \eqref{eq:post_annuitization_value_1}. The functions $V_{\mathrm{int}}(y)$ and $V_{\bar{b}}(y)$ are the value functions corresponding to an interior solution for labor ($b < \bar{b}$) and a corner solution for labor ($b = \bar{b}$), respectively, in the continuation (working) region.
\end{theorem}
\begin{proof} See Appendix~\ref{Proof_of_Theorem_1:_Value_Function}.
\end{proof}

\medskip
\noindent
Theorem~\ref{thm:value_function} establishes the qualitative structure of the value function and identifies the continuation and stopping regions in terms of the wealth-to-habit ratio. The remaining results in this section build on this structure by determining the free-boundary thresholds and deriving the corresponding optimal controls within the continuation region. In particular, we first characterize the annuitization and labor-constraint boundaries, and then derive the optimal consumption, labor, and portfolio policies conditional on non-annuitization.

\begin{proposition}[Optimal Retirement Wealth Threshold]\label{prop:optimal_threshold_y}
Let $V_{\bar{b}}(y)$ and $G(y)$ be the value functions for an agent in the final working phase (with $b=\bar{b}$) and the annuitized state, respectively, as defined in Theorem \ref{thm:value_function}. The optimal retirement wealth threshold, $y^*$, is determined by the \textit{value-matching} and \textit{smooth-pasting} conditions

\begin{align*}
    V_{\bar{b}}(y^*) &= G(y^*), \\
    V_{\bar{b}}'(y^*) &= G'(y^*),
\end{align*}

\noindent The threshold $y^*$ is the unique solution to the non-linear algebraic equation formed by substituting these conditions into the maximized HJB equation for the $V_{\bar{b}}$ region, as detailed in Appendix \ref{app:proof_y_star}.
\end{proposition}
\begin{proof} See Appendix~\ref{app:proof_y_star}.
\end{proof}

\begin{proposition}[Threshold for the Labor Constraint]\label{prop:threshold_y_tilde}
Let $V_{\mathrm{int}}(y)$ and $V_{\bar{b}}(y)$ be the value functions for the interior and corner labor solutions, respectively. The boundary between these states is a unique wealth-to-habit ratio, $\tilde{y}$, determined by the \textit{$C^2$ continuity} conditions
\begin{align*}
    V_{\mathrm{int}}(\tilde{y}) &= V_{\bar{b}}(\tilde{y}), \\
    V_{\mathrm{int}}'(\tilde{y}) &= V_{\bar{b}}'(\tilde{y}), \\
    V_{\mathrm{int}}''(\tilde{y}) &= V_{\bar{b}}''(\tilde{y}).
\end{align*}

\noindent This system of equations implicitly defines the threshold $\tilde{y}$ and the corresponding integration constants ($A_2, B_1, B_2$), as shown in Appendix \ref{Proof_of_Derivation_of_tilde_y}.
\end{proposition}
\begin{proof} See Appendix~\ref{Proof_of_Derivation_of_tilde_y}.
\end{proof}

\medskip
\noindent
We now collect the preceding results to characterize the optimal consumption, labor, and portfolio controls conditional on non-annuitization. The following theorem summarizes the optimal policies within the continuation region \( y < y^* \), taking the value function structure and free-boundary thresholds as given. The global optimality of these controls, together with the optimal stopping rule, is established subsequently via a verification argument.

\begin{theorem}[Optimal Policies in the Continuation Region]\label{thm:optimal_policies_foc}
Let $V(y)$ in Theorem \ref{thm:value_function} be a twice-continuously differentiable solution to the HJB variational inequality in \eqref{eq:hjb_variational_y} for $y < y^*$. The optimal policies are given by the following first-order conditions
\begin{enumerate}
    \item \textit{Optimal Portfolio:} The optimal portfolio weight invested in the risky asset is
    \begin{equation}
        p^*(y) = -\frac{\mu-r}{\sigma^2} \frac{V'(y)}{y V''(y)}. \label{eq:foc_p_main}
    \end{equation}

    \item \textit{Optimal Consumption-to-Habit Ratio:} The optimal ratio $\kappa^*$ is determined by the condition that its marginal utility equals the marginal value of wealth, adjusted for its impact on the state variable dynamics
    \begin{equation}\label{eq:optimal_consumption_policy_foc}
        \frac{\partial u(\kappa^*, b^*)}{\partial \kappa} = V'(y)(1 + \rho y).
    \end{equation}

    \item \textit{Optimal Labor:} For an interior solution where $0 < b^* < \bar{b}$, the optimal labor supply $b^*$ satisfies the condition that the marginal disutility of labor equals its marginal contribution to the value of wealth
    \begin{equation}
        \frac{\partial u(\kappa^*, b^*)}{\partial b} = -w V'(y). \label{eq:foc_b_main}
    \end{equation}
    If this condition cannot be satisfied for any $b \in (0, \bar{b})$, the optimum is a corner solution at the boundary, $b^* = \bar{b}$.

    \item \textit{Marginal Rate of Substitution:} Combining the conditions for optimal consumption and labor for an interior solution yields the optimality condition where the marginal rate of substitution between consumption and leisure equals the wage rate, adjusted by a state-dependent factor
    \begin{equation}
        \frac{\partial u(\kappa^*, b^*) / \partial b}{\partial u(\kappa^*, b^*) / \partial \kappa} = -\frac{w}{1 + \rho y}. \label{eq:MRS_condition_main}
    \end{equation}
\end{enumerate}
\end{theorem}

\begin{corollary}[Actuarially Fair Annuity Rate]\label{cor:actuarially_fair_annuity_rate}
The parameter $k$ is the annuitization value function $G(y)$ in \eqref{eq:objective_functional_habit}, which represents the optimal consumption-to-wealth ratio $c_t/X_t$ upon retirement, can be interpreted as the actuarially fair annuity rate. For an agent of age $t$, this rate is the reciprocal of the annuity present value $\ddot{a}_t$
\begin{equation}
k = \frac{1}{\ddot{a}_t},
\end{equation}
where $\ddot{a}_t$ is the present value of a continuous life annuity paying 1 unit per year, given by
\begin{equation}
\ddot{a}_t = \int_0^\infty e^{-\beta s} \, \px[s]{n}[(\delta)] \, ds = \int_0^\infty e^{-\beta s} \exp\left(-\int_t^{t+s} \delta_u du\right) ds.
\end{equation}
Here, $\beta > 0$ is the subjective discount rate, $\px[s]{n}[(\delta)]$ is the survival probability for an individual aged $t$ to survive $s$ more years, and $\delta_u$ is the age-dependent force of mortality at age $u$.
\end{corollary}

\begin{corollary}[Optimal Annuity Payment]\label{cor:optimal_annuity_payment_rate}

Given optimal policies as in Theorem \ref{thm:optimal_policies_results}, with the optimal retirement wealth threshold $y^*$, the \textit{optimal annuity payment rate} $ \tilde{k}_t^*$ is 
\begin{equation}
k_t^* =
\begin{cases}
0, & \text{if }  y < y^* \quad \text{(working period)}, \\[6pt]
\phi y, & \text{if } y \geq y^* \quad \text{(full retirement period)},
\end{cases}
\end{equation}
where $\phi$ is the endogenous withdrawal rate and satisfies
\begin{equation}
\phi = \frac{\kappa_t^*}{y}= k,
\end{equation}
with $\kappa_t^*$ given by the optimal consumption policy in~\eqref{eq:optimal_consumption_policy_aligned}.
\end{corollary}

\medskip
\noindent
We now state the main verification result, which combines the value function characterization, free-boundary conditions, and continuation-region controls to establish the globally optimal policies and stopping rule.

\begin{theorem}[Optimal Policies]\label{thm:optimal_policies_results}
Assume the value function $V(y)$ from Theorem~\ref{thm:value_function} exists. The optimal policies are functions of the state variable $y_t = X_t/Z_t$ and are given by
\begin{align}
    \kappa_t^* &=
    \begin{cases}
        \left( \dfrac{V'_{\mathrm{int}}(y)(1+\rho y)}{ ((\bar{l} - b_t^*)^{\psi})^{1-\gamma} } \right)^{-1/\gamma}, & \text{if } y < \tilde{y}, \\[2ex]
        \left( \dfrac{V'_{\bar{b}}(y)(1+\rho y)}{ ((\bar{l} - \bar{b})^{\psi})^{1-\gamma} } \right)^{-1/\gamma}, & \text{if } \tilde{y} \leq y < y^*, \\[2ex]
        k y, & \text{if } y \geq y^*,
    \end{cases} \label{eq:optimal_consumption_policy_aligned} \\[2ex]
    b_t^* &=
    \begin{cases}
        \left( \dfrac{1 - \alpha}{w\alpha} \right) \left( \dfrac{k^{1 - \gamma}}{\instdiscount} \right)^{-1/\gamma} y, & \text{if } y < \tilde{y}, \\[2ex]
        \bar{b}, & \text{if } \tilde{y} \leq y < y^*, \\[1ex]
        0, & \text{if } y \geq y^*,
    \end{cases} \label{eq:Optimal_Labor_Supply_main} \\[2ex]
    p_t^* &=
    \begin{cases}
        -\dfrac{\mu-r}{\sigma^2} \dfrac{V'_{\mathrm{int}}(y)}{y V''_{\mathrm{int}}(y)}, & \text{if } y < \tilde{y}, \\[2ex]
        -\dfrac{\mu-r}{\sigma^2} \dfrac{V'_{\bar{b}}(y)}{y V''_{\bar{b}}(y)}, & \text{if } \tilde{y} \leq y < y^*, \\[2ex]
        \dfrac{\mu-r}{\sigma^2 \gamma}, & \text{if } y \geq y^*,
    \end{cases} \label{Optimal_Investment_main_aligned} \\[2ex]
    \tau^* &= \inf\{ t \geq 0 : y_t \geq y^* \}.   \label{Optimal_Stopping_main_aligned}
\end{align}

\noindent where $V'_{\mathrm{int}}(y)$, $V'_{\bar{b}}(y)$, $V''_{\mathrm{int}}(y)$ and $V''_{\bar{b}}(y)$ represent the first and second derivatives of the transformed value function with respect to $y$ in the respective regions defined in Theorem~\ref{thm:value_function}.
\end{theorem}
\begin{proof} See Appendix~\ref{Proof_of_Theorem_1:optimal_policies_results}.
\end{proof}

Section~\ref{Numerical_Implementation} provides an economic interpretation of the optimal policies and illustrates the resulting lifecycle regimes numerically.

\section{Numerical Results, Implementation, and Discussion}\label{Numerical_Implementation}
\noindent We focus on the continuation (working) region ($y < y^*$), where the value function $V(y) > G(y)$, and the stopping (annuitization) region ($y \ge y^*$), where $V(y) = G(y)$. The value $y^*$ represents the optimal retirement wealth threshold for the wealth-to-habit ratio. We numerically implement the optimal results derived in Theorem \ref{thm:value_function} (value function)and Theorem \ref{thm:optimal_policies_results} (optimal policies) and discuss their implications for retirement planning under habit formation.

We set the model parameters based on representative values in the literature (e.g., \cite{birungi2025optimal}, \cite{li2025optimal}, \cite{gao2022optimal}, \cite{chen2021optimal}, and \cite{gerrard2012choosing}), considering an agent aged 60 with an age-dependent force of mortality $\delta_t$. Unless otherwise stated, the \textit{market parameters} include a risk-free rate $r = 0.02$, a risky asset mean return $\mu = 0.07$, and volatility $\sigma = 0.2$, implying a market price of risk $\theta = (\mu-r)/\sigma = 0.25$. The agent's \textit{preference parameters} are set as follows: a subjective discount rate $\beta = 0.03$ (though we may vary this for sensitivity analysis, e.g., in the range $(0.01, 0.055)$), a relative risk aversion coefficient $\gamma = 2.0$, a leisure preference parameter $\psi = 0.5$, and a maximum leisure endowment $\bar{l} = 1.0$. \textit{Labor parameters} include a wage rate $w = 10.0$ and a maximum labor supply $\bar{b} = 0.8$. The \textit{annuity consumption rate parameter} $k$ is typically determined endogenously or via Corollary \ref{cor:actuarially_fair_annuity_rate}. Finally, \textit{mortality} is modeled using an age-dependent force $\delta_t$ (e.g., following the Gompertz law in \eqref{eq:Gompertz_mortality}), which combines with the subjective discount rate to form the effective discount rate $\instdiscount = \beta + \delta_t$.

This section explores the optimal pre- and post-retirement strategies under CRRA utility functions defined over the consumption-to-habit ratio $\kappa_t$ and leisure $(\bar{l}-b_t)$ as described in \eqref{eqn:utility_fun_2}. We analyze how habit formation influences the optimal policies; consumption ($\kappa_t^*$), labor ($b_t^*$), portfolio weight ($p_t^*$), and annuitization timing ($\tau^*$).

\begin{figure}[h!]
    \centering
    \includegraphics[width=0.6\textwidth]{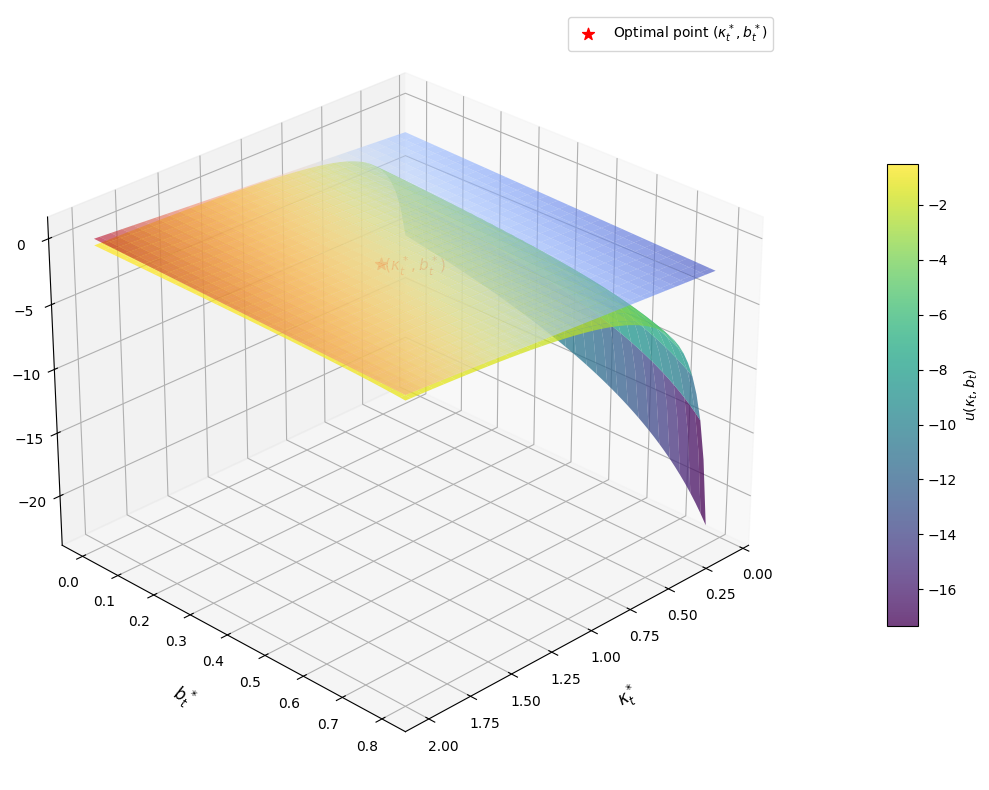}
    \caption{Optimal Policy: Case I (Interior labor solution).}
    \label{fig:case1_interior}
\end{figure}

\begin{figure}[h!]
    \centering
    \includegraphics[width=0.6\textwidth]{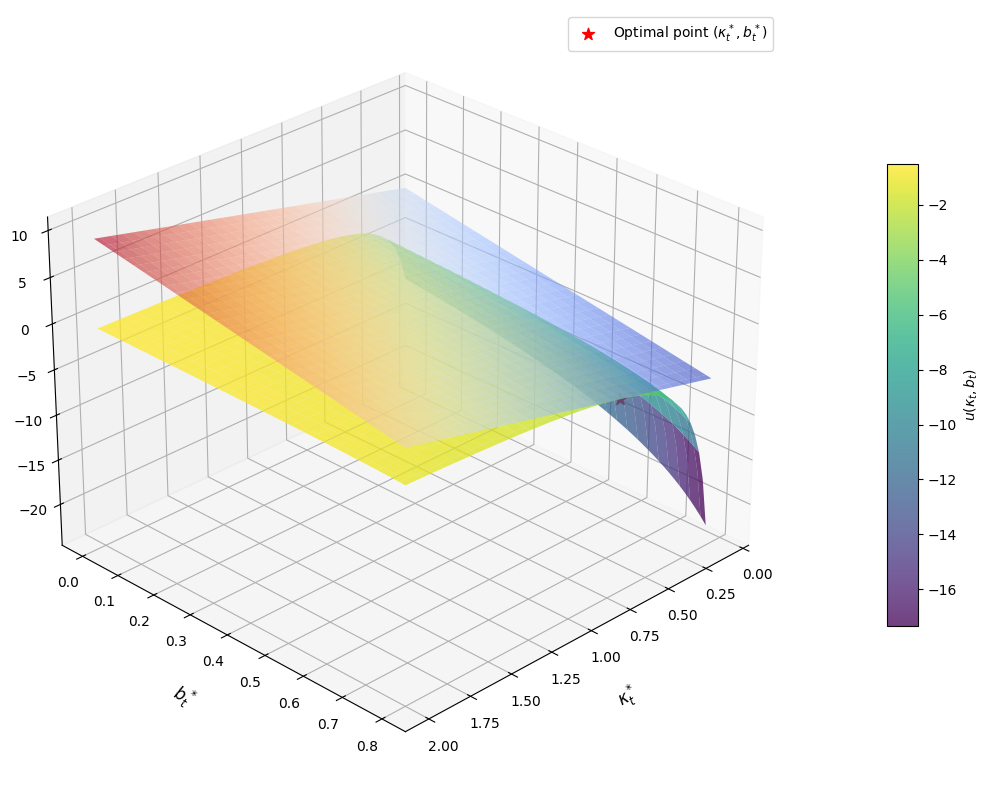}
    \caption{Optimal Policy: Case II (Corner labor solution).}
    \label{fig:case2_corner}
\end{figure}

\noindent \autoref{fig:case1_interior} and \autoref{fig:case2_corner} illustrate the interior and corner labor solutions, respectively, providing a schematic overview of the optimal policies, specifically the consumption-to-habit ratio $\kappa_{t}^*$ and the labor supply $b_{t}^*$, for a given marginal value of wealth $V'(y)$. The tangency point of the interior labor solution occurs within the feasible labor set, indicating that the optimal labor supply $b_{t}^*$ is less than the maximum $\bar{b}$. In the case of the corner labor solution, the agent is constrained by the maximum labor supply $\bar{b}$, which forces the optimal point to the boundary.

\subsection{Policy Characterization}
\noindent
The optimal policies exhibit three economically distinct regimes: a habit-poor region in which labor is supplied defensively to protect consumption standards, a work-to-retire region characterized by maximal labor effort to accelerate annuitization, and a retirement region following irreversible annuitization. In this section, we interpret the behaviour of the optimal policies as functions of the agent's age ($t$) and wealth-to-habit ratio ($y_t$), as visualized in the Figures \ref{fig:kappa_star}, \ref{fig:b_star}, and \ref{fig:p_star}.

The determination of these optimal policies involves plotting the utility surface $u(\kappa_t, b_t)$ and finding the point of tangency with the plane derived from the maximization step of the HJBVI (Hamilton-Jacobi-Bellman Variational Inequality) in equation \eqref{eq:hjb_variational_y}. This optimal tangency point defines $(\kappa_{t}^*, b_{t}^*)$.  In Case I (interior labor solution), this scenario occurs when the wealth-to-habit ratio $y$ lies below the labor-constrained threshold $\tilde{y}$, so that the optimal labor supply is unconstrained and satisfies $b_{t}^{*}<\bar{b}$.  In Case II (corner labor solution) the labor constraint binds for $\tilde{y}\leq y < y^{*}$, in which case the optimal labor supply reaches the boundary $b^{*}_{t} = \bar{b}$

\subsubsection{Optimal consumption to habit ratio ($\kappa_{t}^*$)}

\noindent As illustrated in \autoref{fig:kappa_star}, the optimal consumption-to-habit ratio $\kappa_{t}^*$ exhibits clear dependencies on both state variables. In terms of \textit{wealth effect,} $\kappa_{t}^*$ is strongly increasing with respect to the wealth-to-habit ratio $y_t$. As the agent becomes wealthier relative to their habit ($y_t$ increases), they optimally choose to consume a larger fraction of their habit level, reflecting the standard wealth effect under CRRA preferences. A noticeable change in slope or level occurs at the
 optimal retirement wealth threshold $y^*$, where the policy transitions to the linear rule $\kappa_{t}^* = k y$.

 \begin{figure}[h!]
    \centering
    \includegraphics[width=0.6\textwidth]{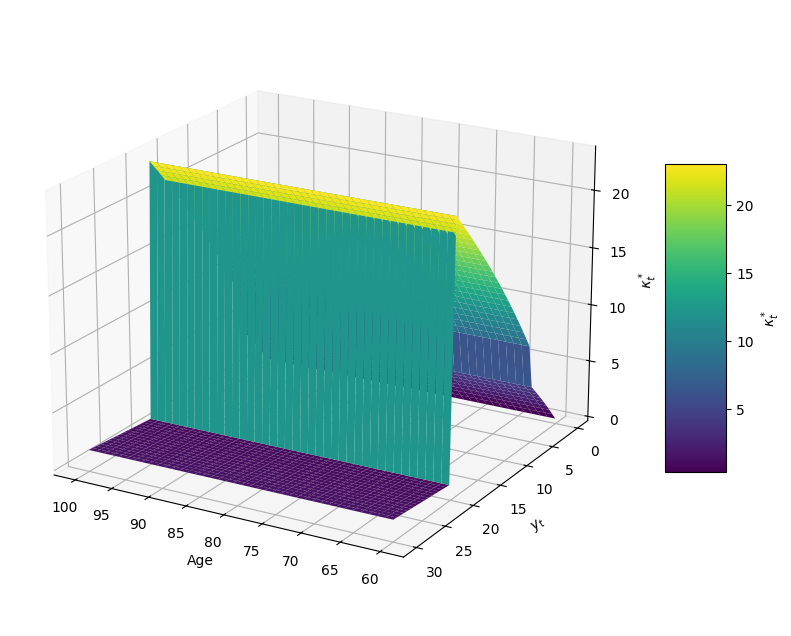}
    \caption{Optimal consumption-to-habit surface as a function of age and the wealth-to-habit ratio.}
    \label{fig:kappa_star}
\end{figure}

\begin{figure}[h!]
    \centering
    \includegraphics[width=0.6\textwidth]{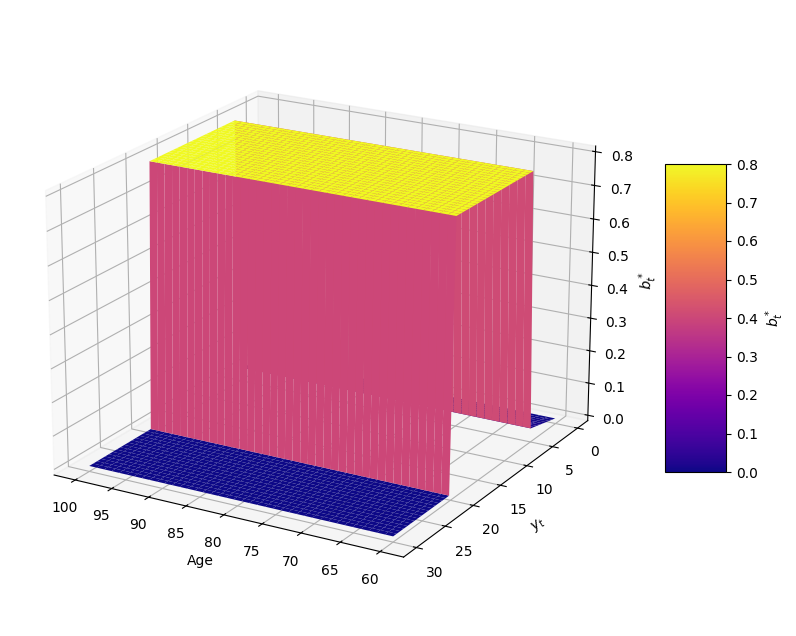}
    \caption{Optimal labor supply surface as a function of age and the wealth-to-habit ratio.}
    \label{fig:b_star}
\end{figure}

\begin{figure}[h!]
    \centering
    \includegraphics[width=0.6\textwidth]{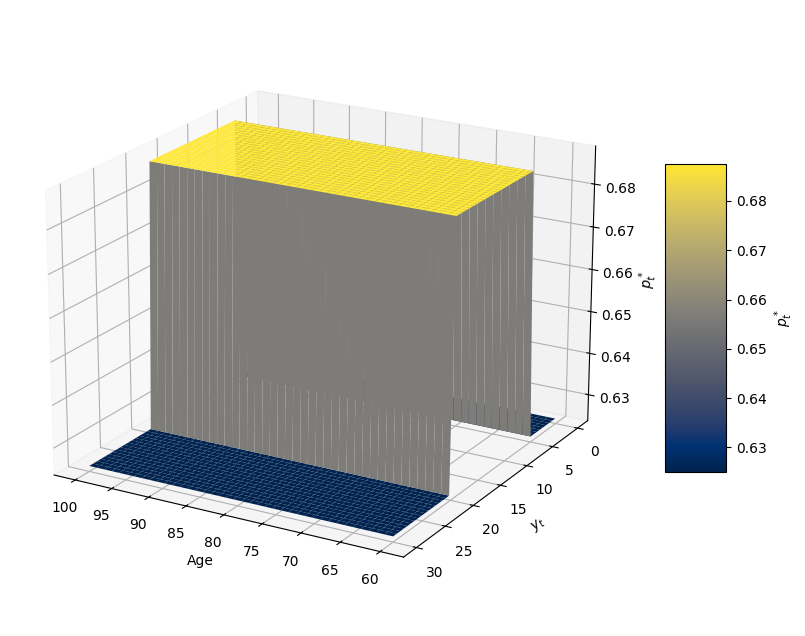}
    \caption{Optimal portfolio weight surface as a function of age and the wealth-to-habit ratio.}
    \label{fig:p_star}
\end{figure}

 In terms of \textit{age effect,} the influence of age is primarily channelled through the effective discount rate $\instdiscount = \beta + \delta_t$. If mortality $\delta_t$ increases significantly with age, $\instdiscount$ rises, making the agent more impatient. This may lead to a slight increase in $\kappa_{t}^*$ with age for a given $y_{t}$, as the agent prioritizes current consumption due to a shorter expected horizon. However, this effect is generally less pronounced than the wealth effect.

\subsubsection{Optimal labor supply ($b_{t}^*$)}

\noindent \autoref{fig:b_star} illustrates the distinct regimes governing the optimal labor supply $b_{t}^*$. In terms of \textit{wealth effect,}  labor supply demonstrates a non-monotonic relationship with the wealth-to-habit ratio $y_{t}$. For low $y$($y < \tilde{y}$), $b_{t}^*$ is positive and determined by the interior solution, potentially increasing with $y_{t}$ as the agent works to accumulate wealth. At the subsistence consumption wealth threshold $\tilde{y}$, $b_{t}^*$ jumps to its upper bound $\bar{b}$, indicating maximal work effort for intermediate levels of $y_{t}$. Finally, upon reaching the optimal retirement wealth threshold $y^*$, $b_{t}^*$ drops discontinuously to zero as the agent retires.
 
Age influences the decision to supply labor. As the agent ages and the force of mortality $\delta_t$ increases, the present value of future wage income decreases. This may reduce the incentive to work, potentially causing a slight decrease in the interior labor supply (for $y < \tilde{y}$) and possibly lowering the thresholds $\tilde{y}$ and $y^*$ over time, making retirement more attractive earlier.

\subsubsection{Optimal portfolio weight ($p_{t}^*$)}

\noindent The optimal fraction of wealth invested in the risky asset, $p_{t}^*$, is shown in \autoref{fig:p_star}. Under the assumption of CRRA utility, the optimal portfolio weight $p_{t}^*$ is largely independent of the wealth-to-habit ratio $y_{t}$ within the continuation region ($y < y^*$). It is approximately equal to the Merton ratio $\mu-r/\sigma^2 \gamma_{\text{eff}}$, (see \cite{koo2013optimal}) where $\gamma_{\text{eff}}$ is the effective risk aversion derived from the value function ($-V'(y)/yV''(y)$). While $\gamma_{\text{eff}}$ might vary slightly between the $V_{\mathrm{int}}$ and $V_{\bar{b}}$ regions, causing small shifts at $\tilde{y}$, the allocation remains relatively stable. In the full retirement region ($y \ge y^*$), the policy becomes constant, $p_{t}^* = \mu-r/\sigma^2 \gamma$, consistent with the standard Merton problem  (see \cite{koo2013optimal} and Remark 3.6.8 in \cite{bertsekas1996stochastic})

 In this model setup, where the primary age-dependent factor is mortality influencing impatience rather than risk aversion or investment opportunities, the optimal portfolio weight $p_{t}^*$ is expected to be largely insensitive to age. The constant relative risk aversion implies a constant allocation fraction throughout the lifecycle.

\subsection{Habit formation influences on optimal policies}
\noindent To understand the impact of habit formation, the labor-leisure choice, and the annuitization option, \autoref{fig:kappa_star_2d}, \autoref{fig:b_star_2d}, and \autoref{fig:p_star_2d} compares the optimal policies derived from our model against a standard CRRA benchmark. The benchmark represents the simple post-annuitization policy (the case where $y \ge y^*$). The comparison is shown across different levels of risk aversion ($\gamma$) as a function of the wealth-to-habit ratio ($y_t$). \autoref{fig:kappa_star_2d} reveal how the pre-annuitization phase or working period with its active labor-leisure trade-off and consumption-smoothing goals cause optimal policies to deviate significantly from this simple benchmark.

\begin{figure}[h!]
    \centering
    \includegraphics[width=0.7\textwidth]{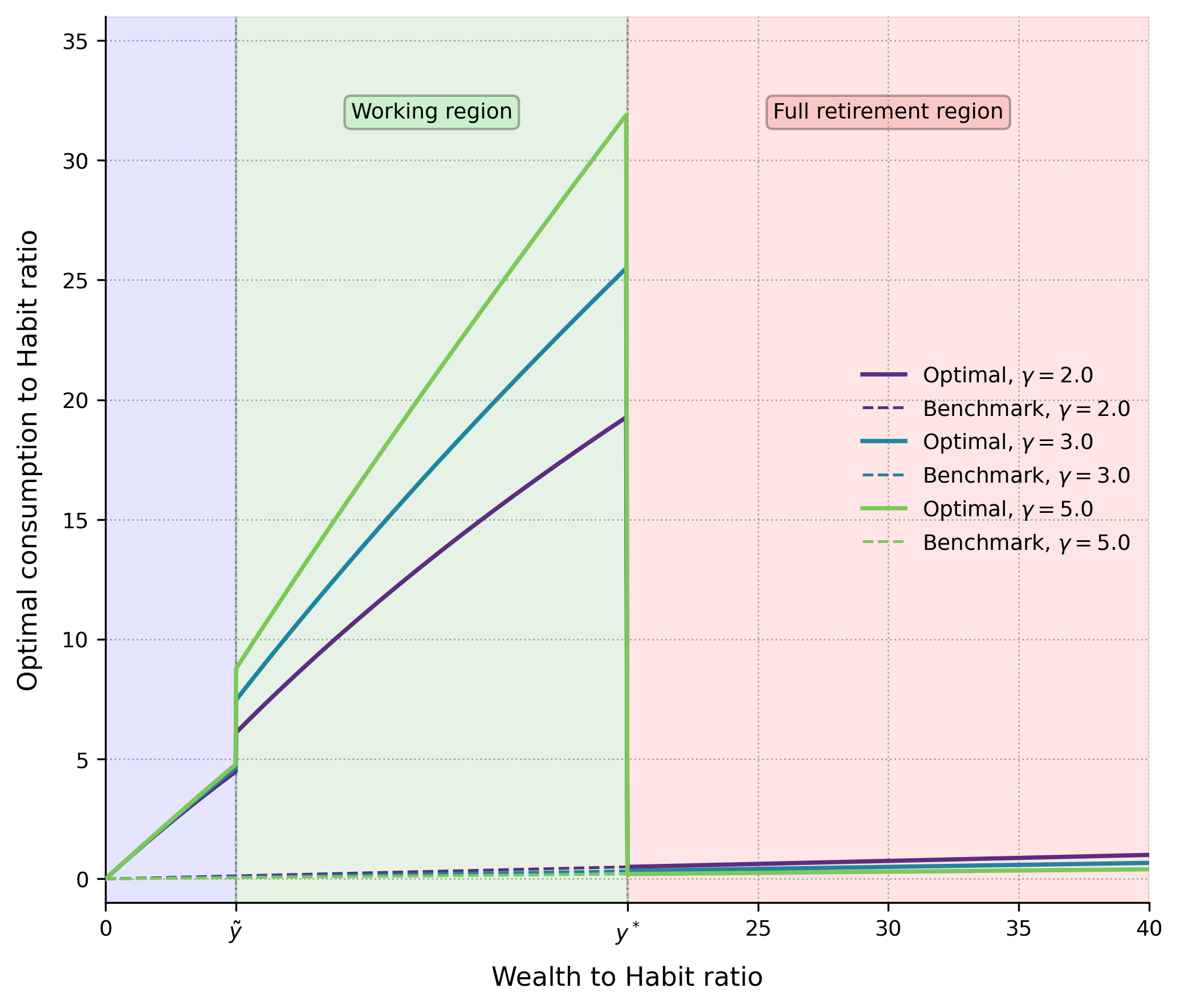}
    \caption{Optimal consumption-to-habit ratio policy under deterministic force of mortality.}
    \label{fig:kappa_star_2d}
\end{figure}

\begin{figure}[h!]
    \centering
    \includegraphics[width=0.7\textwidth]{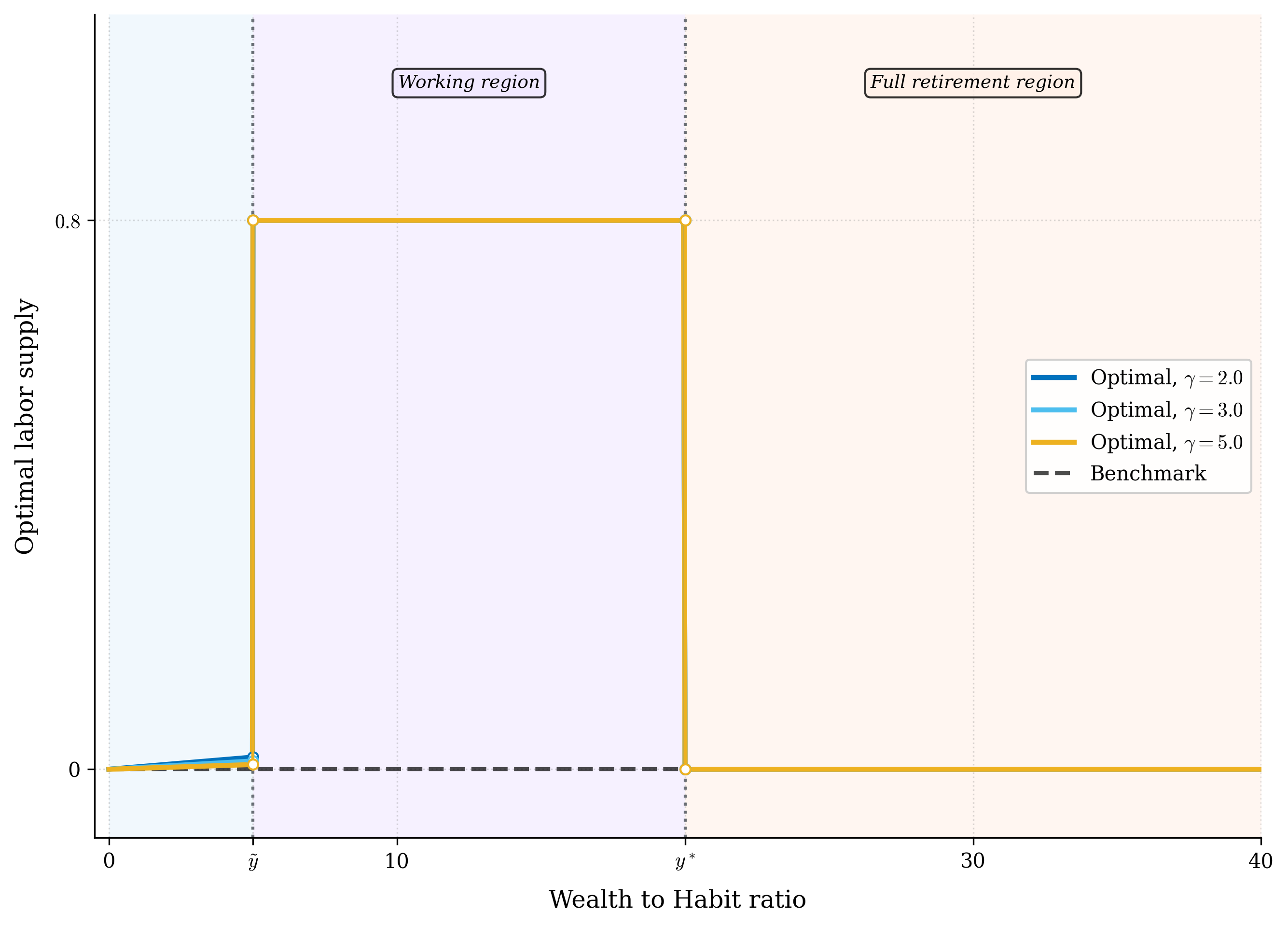}
    \caption{Optimal labor supply policy under a deterministic force of mortality.}
    \label{fig:b_star_2d}
\end{figure}

\begin{figure}[h!]
    \centering
    \includegraphics[width=0.7\textwidth]{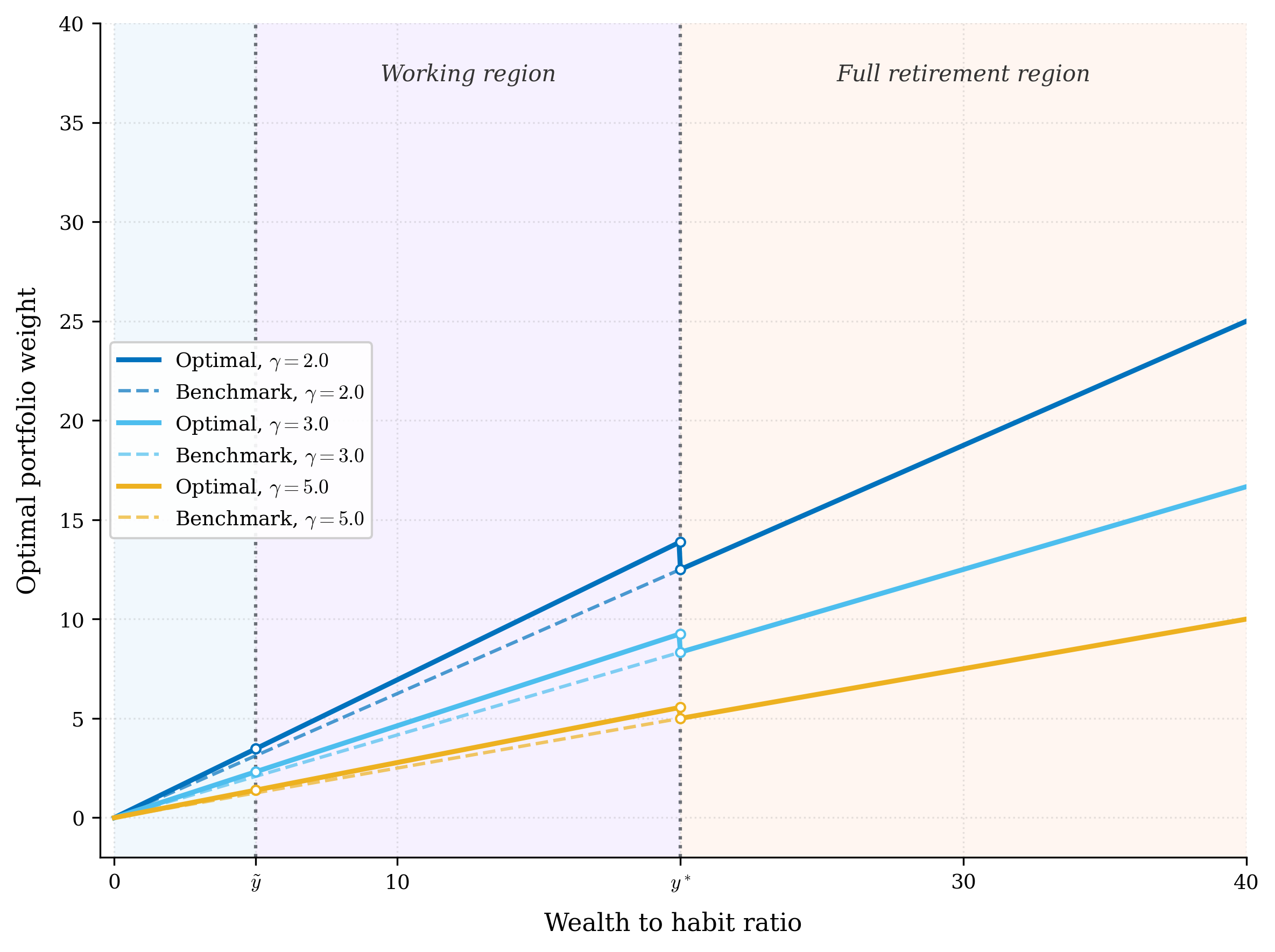}
    \caption{Optimal risky investment strategy under a deterministic force of mortality.}
    \label{fig:p_star_2d}
\end{figure}

\subsubsection{Habit formation influences on  optimal consumption to habit ratio}

\noindent \autoref{fig:kappa_star_2d} shows the direct impact of \textit{consumption habits} and the \textit{labor-leisure trade-off}.

In the \textit{defensive consumption ($y < y^*$),} optimal consumption to habit ratio policies derived from our model differ significantly from the simple benchmark ($\kappa_{t}^* = ky$). At low $y$, the agent is ``habit-poor'' and fights to prevent $\kappa_{t}^*$ from falling too low. They achieve this \textit{defensive consumption} by supplementing wealth with wage income (from $b_{t}^* > 0$). This is the core interplay of \textit{labor being sacrificed to protect the habit.}

Another property of our policies is \textit{discontinuities}: optimal strategies are functions of the state variable $y_t = X_t/Z_t$ as seen in Theorem \ref{thm:optimal_policies_results}. The policy exhibits two sharp changes. The first is a jump at $y = \tilde{y}$, where the labor supply policy switches from interior to the corner solution $\bar{b}$. The second is a kink at $y = y^*$, where the policy smoothly connects to the linear annuitization rule, demonstrating the smooth-pasting condition. And for \textit{risk aversion ($\gamma$)}, higher $\gamma$ makes the agent become more careful and conservative in managing risk or uncertainty. They reduce consumption (i.e., exhibit a lower $\kappa_{t}^*$) across all levels of $y$, thereby accumulating a larger buffer of precautionary savings against future uncertainty.

\subsubsection{Habit formation influences on optimal labor supply}
\noindent \autoref{fig:b_star_2d} illustrates the labor-leisure trade-off as the primary engine of the pre-annuitization phase or working period. The benchmark is $b_{t}^* = 0$ representing the post-annuitization state or full retirement period. In our model, the agent actively uses labor to manage their financial position relative to their habit.

 When \textit{``habit-poor''}, the agent's wealth is low relative to their accustomed lifestyle. To protect their consumption habit, they sacrifice leisure and supply interior labor ($b_{t}^* > 0$). This wage income provides a crucial buffer to supplement wealth.

In the \textit{ ``work-to-retire'' zone,} the agent enters a phase of maximal work ($b_{t}^* = \bar{b}$). They are wealthy enough that the annuitization goal is in sight, but not so wealthy that they can stop. They work as hard as possible to bridge the gap and reach the full retirement threshold $y^*$.

In the \textit{full retirement period}, the agent exercises their \textit{option to annuitize or retire}. They have sufficient wealth to lock in a lifetime consumption stream. Labor immediately drops to zero, and they enter full retirement.

As for \textit{risk aversion ($\gamma$),} higher $\gamma$ (more cautious) agents supply more labor in the interior region ($y < \tilde{y}$). They prefer a certain income from labor over the uncertain returns from financial investment which they hold less of, ( see Figure~\ref{fig:b_star_2d}). They work more to self-insure.

\subsubsection{Habit formation influences on optimal risky investment}
\noindent \autoref{fig:p_star_2d} shows the optimal scaled risky investment, $\hat{\pi}_t^*(y) = \pi_t^*/Z_t$, showing the interplay between financial risk taking and labor income. The slope of the line in this figure represents the optimal portfolio \textit{weight} $p_t^*$. The benchmark policy is given by the line $\hat{\pi}_{bench}^*(y) = p_{bench}^* \cdot y$, where $p_{bench}^* = (\mu-r)/(\sigma^2\gamma)$ is the constant Merton weight that represents the investment policy of a "pure" investor with no labor income. During the working period, the agent's ability to work acts as a ``human capital'' buffer. This labor income allows the agent to take more risk in their financial portfolio. This is why the 'optimal' line ($\gamma=2.0$) in \autoref{fig:p_star_2d} has a steeper slope than its corresponding 'Benchmark' line ($\gamma=2.0$) and the other 'optimal' lines. 
 
There is a \textit{discontinuity at full retirement ($y = y^*$),} at this optimal retirement wealth threshold, the agent retires, and their labor income drops to zero. Their human capital vanishes instantly. To compensate for this loss of a "safe" asset, they immediately become more conservative in their financial portfolio. Their optimal portfolio \textit{weight} $p_t^*$ drops discontinuously to the lower Merton benchmark $p_{bench}^*$. This creates the distinct \textit{downward kink} (a drop in the slope) in the scaled investment plot at $y^*$. As for risk aversion ($\gamma$), a higher $\gamma$ (more risk-averse) means the agent is less willing to bear market risk. This lowers their optimal portfolio weight $p_{t}^*$ in all regions, resulting in a flatter slope for both the 'optimal' and 'benchmark' lines.

\subsection{Mortality risk and survival probability}

\noindent \autoref{fig:survival_age} illustrates the relationship between mortality risk and survival probability in the agent's decision-making process from age 60. It shows that an agent's subjective mortality risk (represented by the modal age of death, $m$) has a significant impact on their optimal policies. The key mechanism is the agent's \textit{effective discount rate}, $\instdiscount = \beta + \delta_t$, which is directly influenced by the force of mortality $\delta_t$.

\begin{figure}[h!]
    \centering
    \includegraphics[width=0.8\textwidth]{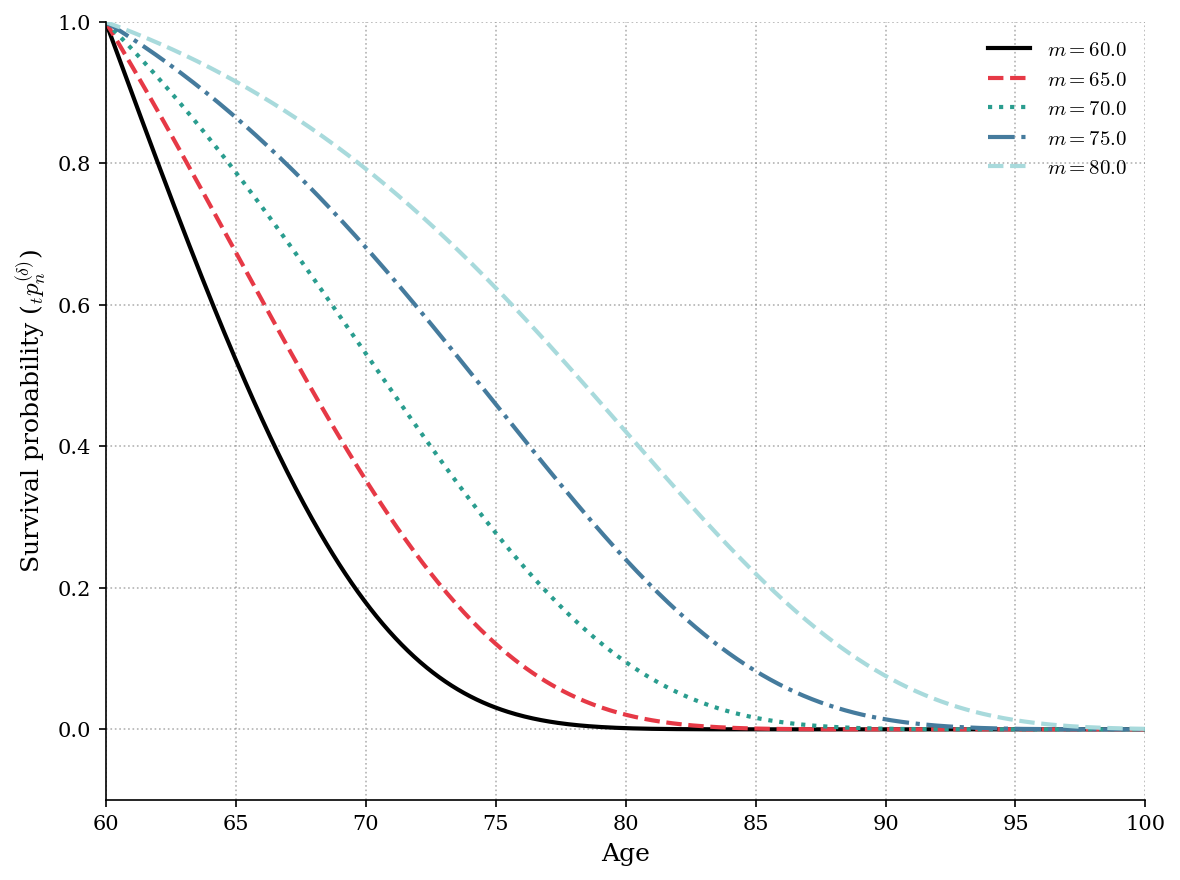} %
    \caption{Agent survival probabilities from age 60 \(\left( \px[t]{60}[(\delta)] \right) \) under the Gompertz law, characterized by the modal age of death ($m$) and the dispersion parameter ($\lambda$). Results are shown for $\lambda=10$ and varying subjective modal ages ($m$).}
    \label{fig:survival_age}
\end{figure}

As seen in \autoref{fig:survival_age},  the \textit{high mortality risk case (e.g., when $m=60$}), the survival probability, $\px[t]{n}[(\delta)]$ drops very quickly. The agent's effective time horizon is very short. This affects the agent's patience, that is because $n=m$, the force of mortality $\delta_t$ is high from the beginning. This causes the effective discount rate $\instdiscount$ to be very \textit{high}, making the agent impatient. For retirement and annuitization, the value of the lifelong income stream option declines substantially. If lifelong is expected to be very short (as the black line($m=60$) shows in \autoref{fig:survival_age}), there is minimal incentive to save for a future that is unlikely to occur. This implies the optimal retirement wealth threshold $y^*$ will be \textit{very low}. In terms of interplay of mortality risk with labor and consumption, because the future is so uncertain (per the $m=60$ curve) and the annuitization prize has low value, the agent will \textit{work less} (lower $b_{t}^*$) and \textit{consume more} (higher $\kappa_{t}^*$) \textit{today}. The labor-leisure trade-off shifts heavily towards current consumption and leisure.

There is low mortality risk (e.g., the $m=80$). This shows that the survival probability, $\px[t]{n}[(\delta)]$, declines very slowly. The agent perceives a long and probable future, as they are 20 years away from their modal age of death. Again, this affects the agent's patience; that is, the force of mortality, $\delta_t$, is low for many years. This keeps the effective discount rate $\instdiscount$ low, making the agent "patient." Annuitization becomes highly valuable since it represents a guaranteed, high-value income stream for decades. This creates a powerful incentive to work and save to achieve this goal. Regarding the interplay with labor and consumption, the agent is highly motivated to reach the optimal retirement wealth threshold $y^*$. This means the agent will \textit{work more} (higher $b_{t}^*$) and \textit{save more by consuming less} (lower $\kappa_{t}^*$) during their working years. The labor-leisure trade-off shifts towards labor to secure the valuable retirement prize.

In summary, the mortality risk and survival probability curve dictates the agent's motivation. A steep curve (like $m=60$) breaks the incentive to save, leading to less work, higher current consumption, and a lower annuitization target. A flat curve (like $m=80$) strengthens the incentive to save, leading to more work, more saving, and a higher, more valuable annuitization goal.

\subsubsection{Mortality risk  after retirement}

\noindent \autoref{fig:survival_age} and Table \ref{tab:subjective_premium} illustrate the effect of subjective mortality beliefs on the agent's annuitization decision. The values in Table \ref{tab:subjective_premium} are calculated as a \textit{Normalized Premium Ratio (NPR)}. This ratio compares the agent's perceived value of the annuity (Subjective Premium $\tilde{P}$) to the insurer's market price (Objective Premium $P_{\text{obj}}$).

\begin{table}[htbp]
\centering
\caption{Normalized subjective annuity premium $\tilde{P}_0^A$ for a 60-year-old agent, 
given an objective premium $P_0^A = 1$ based on an insurer’s modal age $m = 80.0$.}
\label{tab:subjective_premium}
\vspace{0.3em}
\begin{tabular}{lcc}
\toprule
Parameters & Subjective modal age $\tilde{m}$ & Normalized annuity premium $\tilde{P}_0^A$ \\
\midrule
Case 1 & 60.0 & 0.3642 \\
Case 2 & 65.0 & 0.4913 \\
Case 3 & 70.0 & 0.6385 \\
Case 4 & 75.0 & 0.8066 \\
Case 5 & 80.0 & 1.0000 \\
\bottomrule
\end{tabular}

\vspace{0.5em}
\footnotesize{
*The second column represents the subjective modal age $\tilde{m}$ for which the normalized annuity premium $\tilde{P}_0^A$ is computed.  
**Basic parameters:** $n = 60$, $\lambda = 10$, $r = 0.02$, and $P_0^A = 1$. 
}
\end{table}

\noindent The premium $P$ for a \$1 continuous annuity is its actuarial present value, defined by the agent's current age $n$, modal age $m$, dispersion $\lambda$, and the risk-free rate $r$
\begin{equation}
    P(m) = \int_0^\infty e^{-rt} \cdot \px[t]{n}[(\delta)](m)dt
\end{equation}
where $\px[t]{n}[(\delta)](m)$ is the Gompertz survival probability from equation \eqref{eq:age_dependent_mortality}. The NPR is the ratio of the premium calculated using the agent's subjective modal age, $\tilde{m}$, to the premium calculated using the insurer's objective modal age, $m_{\text{obj}}$
\begin{equation}
    \text{NPR}(\tilde{m}) = \frac{\tilde{P}}{P_{\text{obj}}} = \frac{P(\tilde{m})}{P(m_{\text{obj}})} = \frac{\int_0^\infty e^{-rt} \cdot \px[t]{n}[(\delta)](m) \, dt}{\int_0^\infty e^{-rt} \cdot 
    \px[t]{n}[(\delta)](m_{\text{obj}}) dt}
\end{equation}

\noindent  We classify agents as pessimistic retirees, that is, \textit{pessimistic retirees (those with $\tilde{m} < 80.0$)}, a group of individuals who believe they will die sooner than the "objective" insurer's estimate of $m=80$. Extreme pessimism ($\tilde{m} = 60.0$), as shown in \autoref{fig:survival_age} with $\tilde{m} = 60.0$), this agent is already at their modal age of death. Their survival probability drops sharply.
Table \ref{tab:subjective_premium} shows their subjective premium is only \textit{36.4\%} of the market price. For them, the annuity is perceived as an extremely bad deal. Moderate pessimism ($\tilde{m} = 75.0$) refers to agents who expect to die 5 years earlier than the benchmark, still only valuing the annuity at 80.7\% of its price.

The economic implication of all pessimistic agents, $\text{NPR}< 1$, is that they perceive the value of annuitization as low.  This disincentivizes annuitization, which in our model translates to a higher annuitization threshold  $y^{*}$ , or, in extreme cases, to effectively avoiding annuitization altogether.

On the other hand, the \textit{neutral retiree ($\tilde{m} = 80.0$),} is an agent's beliefs (\textit{illustrated by $\tilde{m} = 80.0$} in \autoref{fig:survival_age}) align perfectly with the insurer's. They perceive the annuity as fairly priced ($\text{NPR} = 1.0000$). They will annuitize at the "baseline" threshold $y^*$ calculated by the model, balancing current work and consumption against the fair value of a future income stream.

\section{Conclusion}\label{Conclusion_Recommendations}

\noindent
This paper studied a coupled stochastic optimal control and optimal stopping problem in which an agent jointly manages consumption relative to habit, flexible labor supply, portfolio choice, and an irreversible annuitization decision under age-dependent mortality. By formulating the problem as a Hamilton-Jacobi-Bellman variational inequality and exploiting homotheticity, we derived tractable semi-analytical solutions characterizing optimal policies across the lifecycle.

The model generates a rich sequence of retirement dynamics. When wealth is low relative to habit, labor is used defensively to protect consumption standards. As wealth accumulates, agents enter a work-to-retire phase in which labor is supplied at its maximum level to accelerate access to retirement and annuitization. Human capital acts as a stabilizing asset during working life, justifying more aggressive portfolio risk-taking, followed by an abrupt de-risking at annuitization when this asset vanishes.

Subjective mortality beliefs play a central role in shaping annuitization behavior. Agents with pessimistic longevity beliefs rationally perceive annuities as unattractive, leading them to delay or avoid annuitization since it requires higher retirement thresholds. In this way, the model provides a behavior-based explanation for low annuity demand without appealing to market frictions or irrationality.

Overall, the numerical analysis highlights how habit formation and endogenous labor supply fundamentally reshape retirement dynamics, portfolio risk-taking, and annuitization decisions. From a practical perspective, the results suggest that pre-retirement financial advice should account explicitly for labor flexibility and subjective longevity beliefs, supporting aggressive pre-retirement investment strategies followed by sharp de-risking at retirement.

Several extensions merit future investigation. Incorporating stochastic labor income or stochastic market volatility would introduce additional hedging motives. Allowing for partial annuitization rather than an all-or-nothing decision would add further realism. Modeling health status as a stochastic factor jointly affecting mortality and leisure preferences would also be a valuable extension, but is beyond the scope of this paper.

As populations age and defined-contribution retirement systems become increasingly prevalent, models that bridge psychological realism and financial optimization are essential for understanding retirement behavior and informing long-term financial planning.

\bibliographystyle{plainnat}
\bibliography{reference.bib}

\appendix
\section{Appendix: Detailed Proofs of Theorems}
\noindent This appendix provides detailed derivations for the main results of the paper.

\subsection{Proof of Proposition \ref{thm:Dynamic_programming_principle_habit}: Dynamic Programming Principle (DPP)}\label{app:DPP_Proof}

\noindent The proof of the Dynamic Programming Principle follows standard arguments for combined control and stopping problems and is included for completeness.  

\begin{proof}
Let the objective functional for an admissible strategy $(\pi, c, b)$ and a stopping time $\tau$ be denoted by $J(t, x, z; \pi, c, b, \tau)$. The value function is defined as the supremum over the set of all admissible strategies $\mathcal{A}$, such that
$$
V(t, x, z) = \sup_{(\pi, c, b, \tau) \in \mathcal{A}} J(t, x, z; \pi, c, b, \tau).
$$
To prove the equality in \eqref{eq:dpp}, we will establish two inequalities. Let the right-hand side of \eqref{eq:dpp} be denoted by $\text{RHS}$.

\vspace{1em}
\noindent\textit{Part 1: Show that 
\begin{equation} \label{eq:dpp_greater}
\begin{split}
    V(t, x, z) \ge \max \Biggl\{ G(x), \sup_{(\pi, c, b)} \mathbb{E}_{t,x,z} \biggl[ &\int_t^{t+\Delta t} e^{-\int_t^s (\beta + \delta_u) du} U(c_s,b_s) ds \\
    &+ e^{-\int_t^{t+\Delta t} (\beta + \delta_u) du} V(t+\Delta t, X_{t+\Delta t}, Z_{t+\Delta t}) \biggr] \Biggr\}.
\end{split}
\end{equation}
}
\noindent We write equation \eqref{eq:dpp_greater} as $V(t, x, z) \ge \text{RHS}$. First, we show that $V(t, x, z)$ is greater than or equal to each term inside the $\max(\cdot, \cdot)$ operator.
\begin{enumerate}
    \item The strategy of stopping immediately at time $\tau=t$ is admissible. This yields a value of $G(x)$. Since the value function is the supremum over all admissible strategies, it must be that $V(t, x, z) \ge G(x)$.

    \item Next, consider an admissible strategy that applies an arbitrary control policy $(\pi, c, b)$ on the interval $[t, t+\Delta t]$ and then proceeds optimally from time $t+\Delta t$ onwards. The value of this combined strategy is
    \begin{equation}\label{eq:combined_strat_val}
    \begin{split}
        \mathbb{E}_{t,x,z}\biggl[ &\int_t^{t+\Delta t} e^{-\int_t^s (\beta + \delta_u) du} U(c_s,b_s) ds \\
        &\quad + e^{-\int_t^{t+\Delta t} (\beta + \delta_u) du} V(t+\Delta t, X_{t+\Delta t}, Z_{t+\Delta t}) \biggr].
    \end{split}
    \end{equation}
\end{enumerate}
Because the value function $V(t, x, z)$ is the supremum over all admissible strategies, its value must be greater than or equal to the value of this specific strategy. This holds for any choice of $(\pi, c, b)$ on $[t, t+\Delta t]$, and thus it must also hold for the supremum over all such choices. Therefore,
    \begin{align*}
 &V(t, x, z) \ge \sup_{(\pi, c, b)} \mathbb{E}_{t,x,z} \biggl[ \int_t^{t+\Delta t} e^{-\int_t^s (\beta + \delta_u) du} U(c_s,b_s) ds %
          + e^{-\int_t^{t+\Delta t} (\beta + \delta_u) du} V(t+\Delta t, X_{t+\Delta t}, Z_{t+\Delta t}) \biggr].
    \end{align*}
Since $V(t, x, z)$ is greater than or equal to both arguments of the $\max(\cdot, \cdot)$ operator, it must be greater than or equal to their maximum. Thus, we have established that $V(t, x, z) \ge \text{RHS}$.

\vspace{1em}
\noindent\textit{Part 2: Show that $V(t, x, z) \le \text{RHS}$.} For any given $\varepsilon > 0$, the definition of the value function guarantees the existence of an $\varepsilon$-optimal strategy $(\pi^\varepsilon, c^\varepsilon, b^\varepsilon, \tau^\varepsilon)$ such that
\begin{align}\label{eq:eps_optimal}
    V(t, x, z) \le J(t, x, z; \pi^\varepsilon, c^\varepsilon, b^\varepsilon, \tau^\varepsilon) + \varepsilon.
\end{align}
We consider two possibilities for the stopping time $\tau^{\varepsilon}$:
\begin{itemize}
    \item[1.] If $\tau^{\varepsilon} = t$, the strategy is to stop immediately, so
    \[
        J(t, x, z; \pi^{\varepsilon}, c^{\varepsilon}, b^{\varepsilon}, t) = G(x).
    \]
    Substituting this into \eqref{eq:eps_optimal} gives
    $V(t, x, z) \le G(x) + \varepsilon$.
    By definition, $G(x) \le \text{RHS}$, which implies
    $V(t, x, z) \le \text{RHS} + \varepsilon$.
    
    \item[2.] If $\tau^{\varepsilon} > t$, we decompose the objective function.  By the law of iterated expectations and by the definition of the value function, we have
\end{itemize}
{\allowdisplaybreaks
\begin{align}
    J&(t, x, z; \pi^{\varepsilon}, c^{\varepsilon}, b^{\varepsilon}, \tau^{\varepsilon})
    = \mathbb{E}_{t,x,z}\Biggl[
        \int_t^{\tau^{\varepsilon}} 
            e^{-\int_t^s (\beta + \delta_u)\,du} 
            U(c_s^{\varepsilon}, b_s^{\varepsilon})\,ds
        + e^{-\int_t^{\tau^{\varepsilon}} (\beta + \delta_u)\,du} 
            G(X_{\tau^{\varepsilon}})
    \Biggr] \nonumber \\[0.5em]
    &= \mathbb{E}_{t,x,z}\Biggl[
        \int_t^{t+\Delta t} 
            e^{-\int_t^s (\beta + \delta_u)\,du} 
            U(c_s^{\varepsilon}, b_s^{\varepsilon})\,ds %
        + e^{-\int_t^{t+\Delta t} (\beta + \delta_u)\,du} 
            J(t+\Delta t, X_{t+\Delta t}, Z_{t+\Delta t};
            \pi^{\varepsilon}, c^{\varepsilon}, b^{\varepsilon}, \tau^{\varepsilon})
    \Biggr] \nonumber \\[0.5em]
    &\le \mathbb{E}_{t,x,z}\Biggl[
        \int_t^{t+\Delta t} 
            e^{-\int_t^s (\beta + \delta_u)\,du} 
            U(c_s^{\varepsilon}, b_s^{\varepsilon})\,ds
        + e^{-\int_t^{t+\Delta t} (\beta + \delta_u)\,du}  V(t+\Delta t, X_{t+\Delta t}, Z_{t+\Delta t})
    \Biggr]. \label{eq:J_less_than_E}
\end{align}
} %
\noindent The right-hand side of inequality~\eqref{eq:J_less_than_E} is the expected value from following the specific policy
$(\pi^{\varepsilon}, c^{\varepsilon}, b^{\varepsilon})$ over the interval $[t, t+\Delta t]$.
This value cannot exceed the supremum taken over all admissible policies on that interval. Therefore,
    \begin{align}\label{eq:E_less_than_sup}
    \begin{split}
        &\mathbb{E}_{t,x,z}\Biggl[ \int_t^{t+\Delta t} e^{-\int_t^s (\beta+\delta_u) du} U(c_s^\varepsilon,b_s^\varepsilon) ds + e^{-\int_t^{t+\Delta t} (\beta+\delta_u) du} V(t+\Delta t, X_{t+\Delta t}, Z_{t+\Delta t}) \Biggr] \\
        &\quad\le \sup_{(\pi, c, b)} \mathbb{E}_{t,x,z} \biggl[ \int_t^{t+\Delta t} e^{-\int_t^s (\beta + \delta_u) du} U(c_s,b_s) ds + e^{-\int_t^{t+\Delta t} (\beta + \delta_u) du} V(t+\Delta t, X_{t+\Delta t}, Z_{t+\Delta t}) \biggr] \\
        &\quad\le \text{RHS}.
    \end{split}
    \end{align}
   
    \noindent Combining \eqref{eq:eps_optimal}, \eqref{eq:J_less_than_E}, and \eqref{eq:E_less_than_sup}, we find that
    $$
    V(t, x, z) \le J(t, x, z; \pi^\varepsilon, c^\varepsilon, b^\varepsilon, \tau^\varepsilon) + \varepsilon \le \text{RHS} + \varepsilon.
    $$

    \noindent In both cases for $\tau^\varepsilon$, we conclude that $V(t, x, z) \le \text{RHS} + \varepsilon$. Since this holds for any arbitrarily small $\varepsilon > 0$, we must have $V(t, x, z) \le \text{RHS}$. Combining the results from Part 1 and Part 2, we have proven the equality in \eqref{eq:dpp}.
\end{proof}

\subsection{Proof of Theorem \ref{thm:value_function} (Value Function)}\label{Proof_of_Theorem_1:_Value_Function}

\noindent We analyze the pre-annuitization problem where the agent has not yet reached the optimal stopping boundary, i.e., $y < y^*$. In this continuation region, the value function $V(y)$ is greater than the annuitization value $G(y)$ and must satisfy the Hamilton-Jacobi-Bellman (HJB) equation. We analyze the problem by partitioning the continuation region based on the optimal labor supply $b^*$.\\

\begin{proof}
\noindent The stationary value function $V(y)$ satisfies the HJB equation
\begin{equation}\label{eq:HJB_pre_proof_detail}
\instdiscount V(y) = \sup_{p, \kappa, b} \left[ u_1(\kappa, b) + \mathcal{L}V(y) \right],
\end{equation}
where $u_1(\kappa,b) = \frac{(\kappa(\bar{l} - b)^{\psi})^{1-\gamma}}{1-\gamma}$ and the generator is
\[
\mathcal{L}V(y) = V'(y) \left[ (r+\rho)y + p y(\mu-r) - \kappa(1+\rho y) + w b \right] + \frac{1}{2} V''(y) \sigma^2 p^2 y^2.
\]

\vspace{1em}
\noindent
\textbf{Case I: $y < \tilde{y}$ (Interior Labor Supply $0 \le b^* < \bar{b}$).}
\vspace{0.5em}
In this region, all controls are interior. We substitute the optimal portfolio $p^*(y) = -\frac{\mu-r}{\sigma^2} \frac{V'(y)}{y V''(y)}$ into the HJB equation \eqref{eq:HJB_pre_proof_detail} to get the maximized HJB for $V(y) = V_{\mathrm{int}}(y)$
\begin{equation}\label{eq:HJB_maximized_int}
\instdiscount V_{\mathrm{int}}(y) = u_1(\kappa^*, b^*) + V'_{\mathrm{int}}(y) \left[ (r+\rho)y - \kappa^*(1+\rho y) + w b^* \right] - \frac{1}{2}\left(\frac{\mu-r}{\sigma}\right)^2 \frac{(V'_{\mathrm{int}}(y))^2}{V''_{\mathrm{int}}(y)}.
\end{equation}
To solve this non-linear PDE, we use a dual method. We assume the optimal policy $\kappa^*=K_{\mathrm{int}}(y)$ is invertible, with inverse $y=Y_{\mathrm{int}}(\kappa)$. The FOCs from Theorem \ref{thm:optimal_policies_foc} relate the value function's derivatives to the optimal controls
\begin{align}
V'_{\mathrm{int}}(y) &= \frac{1}{1+\rho y} \frac{\partial u_1(\kappa, b^*)}{\partial \kappa}, \label{eq:V_prime_int_y_kappa} \\
V''_{\mathrm{int}}(y) &= \frac{d}{dy} \left( V'_{\mathrm{int}}(y) \right). \label{eq:V_double_prime_int_y_kappa}
\end{align}
Substituting these back into the maximized HJB \eqref{eq:HJB_maximized_int} and differentiating the entire equation with respect to $\kappa$ transforms the problem into a second-order linear ordinary differential equation (ODE) for the inverse function $Y_{\mathrm{int}}(\kappa)$
\begin{equation}\label{eq:second_order_ODE_detail_y}
\mathcal{A}_{\mathrm{int}} \kappa^2 Y''_{\mathrm{int}}(\kappa) + \mathcal{B}_{\mathrm{int}}(\kappa, Y) \kappa Y'_{\mathrm{int}}(\kappa) + \mathcal{C}_{\mathrm{int}}(\kappa, Y) Y_{\mathrm{int}}(\kappa) = \mathcal{D}_{\mathrm{int}}(\kappa, Y),
\end{equation}
where the coefficients are functions of model parameters. The general solution is of the form $Y_{\mathrm{int}}(\kappa) = A_1 \kappa^{m_1} + A_2 \kappa^{m_2} + Y_{p}(\kappa)$. Economic boundary conditions typically ensure $A_1=0$, leaving a solution dependent on one constant, $A_2$.

\vspace{1em}
\noindent
\textbf{Case II: $\tilde{y} \le y < y^*$ (Corner Labor Supply $b=\bar{b}$)}
\vspace{0.5em}
Here, $b^*=\bar{b}$. The maximized HJB equation for $V(y) = V_{\bar{b}}(y)$ is:
\begin{equation}\label{eq:HJB_maximized_bar_b}
\instdiscount V_{\bar{b}}(y) = u_1(\kappa^*, \bar{b}) + V'_{\bar{b}}(y) \left[ (r+\rho)y - \kappa^*(1+\rho y) + w \bar{b} \right] - \frac{1}{2}\left(\frac{\mu-r}{\sigma}\right)^2 \frac{(V'_{\bar{b}}(y))^2}{V''_{\bar{b}}(y)}.
\end{equation}
The procedure is analogous to Case I. We use the FOC for $\kappa^*$ with $b$ fixed at $\bar{b}$ to relate $V'_{\bar{b}}(y)$ and $V''_{\bar{b}}(y)$ to $\kappa$. This again transforms the HJB into a second-order linear ODE for the inverse function $y=Y_{\bar{b}}(\kappa)$
\begin{equation}\label{eq:second_order_ODE_post_detail_y}
\mathcal{A}_{\bar{b}} \kappa^2 Y''_{\bar{b}}(\kappa) + \mathcal{B}_{\bar{b}}(\kappa, Y) \kappa Y'_{\bar{b}}(\kappa) + \mathcal{C}_{\bar{b}}(\kappa, Y) Y_{\bar{b}}(\kappa) = \mathcal{D}_{\bar{b}}(\kappa, Y).
\end{equation}
The general solution is $Y_{\bar{b}}(\kappa, B_1, B_2) = B_1 \kappa^{m'_1} + B_2 \kappa^{m'_2} + Y_{p, \bar{b}}(\kappa)$.

\vspace{1em}
\noindent
\textbf{Annuitization Region: $y \ge y^*$}
\vspace{0.5em}
In this region, it is optimal to stop. The value function is equal to the transformed annuitization value: $V(y)=G(y)=\frac{(ky)^{1-\gamma}}{\instdiscount (1-\gamma)}$. The habit level $Z$ is factored out via the homothetic transformation.

\vspace{1em}
\noindent
The constants of integration ($A_2, B_1, B_2$) and the thresholds ($\tilde{y}, y^*$) are determined by imposing $C^2$ continuity conditions at the boundaries $y=\tilde{y}$ and $y=y^*$. This system of value-matching, smooth-pasting, and super-contact equations uniquely determines the constants and completes the proof.
\end{proof}

\subsection{Proof of Theorem \ref{thm:optimal_policies_results} (Optimal Policies)}\label{Proof_of_Theorem_1:optimal_policies_results}

\noindent
This appendix provides the verification argument establishing that the candidate value function and policies characterized in Theorems~\ref{thm:value_function} and~\ref{thm:optimal_policies_foc} are globally optimal and satisfy the Hamilton--Jacobi--Bellman variational inequality.

\begin{proof}

The optimal policies are derived by applying the FOCs from Theorem \ref{thm:optimal_policies_foc} to the value function $V(y)$ in each of the three regions defined by $\tilde{y}$ and $y^*$.

\begin{enumerate}
\item \textbf{Optimal Investment $p_t^*$ (Portfolio Weight):}
The FOC is $p^*(y) = -\frac{\mu-r}{\sigma^2}\frac{V'(y)}{yV''(y)}$.
\begin{itemize}
    \item For $y<\tilde{y}$: $p_t^* = -\frac{\mu-r}{\sigma^2} \frac{V'_{\mathrm{int}}(y)}{y V''_{\mathrm{int}}(y)}$.
    \item For $\tilde{y} \le y < y^*$: $p_t^* = -\frac{\mu-r}{\sigma^2} \frac{V'_{\bar{b}}(y)}{y V''_{\bar{b}}(y)}$.
    \item For $y \ge y^*$: Here $V(y)=G(y)=\frac{(ky)^{1-\gamma}}{\instdiscount(1-\gamma)}$. The derivatives are:
    \[ G'(y) = \frac{k^{1-\gamma}}{\instdiscount} y^{-\gamma}, \qquad G''(y) = -\frac{\gamma k^{1-\gamma}}{\instdiscount} y^{-\gamma-1}. \]
    Substituting into the FOC for the portfolio weight gives a constant
    \[ p^*(y) = -\frac{\mu-r}{\sigma^2} \frac{G'(y)}{yG''(y)} = -\frac{\mu-r}{\sigma^2} \frac{ \frac{k^{1-\gamma}}{\instdiscount} y^{-\gamma} }{ y(-\frac{\gamma k^{1-\gamma}}{\instdiscount} y^{-\gamma-1}) } = \frac{\mu-r}{\sigma^2 \gamma}. \]
    The total dollar amount in the risky asset is $\pi_t^* = p_t^* X_t = \left(\frac{\mu-r}{\sigma^2 \gamma}\right) y_t Z_t$.
\end{itemize}

\item \textbf{Optimal Consumption-to-Habit Ratio $\kappa_t^*$:}
\begin{itemize}
    \item For $y < y^*$: The FOC is $\frac{\partial u_1}{\partial \kappa} = V'(y)(1 + \rho y)$. For our utility, the marginal utility is $\frac{\partial u_1}{\partial \kappa} = \kappa^{-\gamma}((\bar{l}-b)^{\psi})^{1-\gamma}$. Inverting this for $\kappa^*$ yields
    \[ \kappa^*(y) = \left( \frac{V'(y)(1+\rho y)}{((\bar{l} - b^*(y))^{\psi})^{1-\gamma}} \right)^{-1/\gamma}. \]
    Applying this to the $V_{\mathrm{int}}$ and $V_{\bar{b}}$ regions confirms the results. The total consumption is $c_t^* = \kappa_t^* Z_t$.
    \item For $y \ge y^*$: The agent is annuitized. The optimal consumption rule from a Merton-style problem is $c_t^* = k X_t$. In our framework, this translates directly to the ratios: $c_t^* = \kappa_t^* Z_t$ and $X_t = y_t Z_t$. Therefore, $\kappa_t^* Z_t = k y_t Z_t$, which implies
    \[ \kappa_t^* = k y_t. \]
\end{itemize}

\item \textbf{Optimal Labor Supply $b_t^*$:}
\begin{itemize}
    \item For $y<\tilde{y}$: Labor supply is interior, derived from the MRS condition. The explicit solution provided in the theorem arises from solving the system of FOCs for a utility function with specific properties relating consumption and leisure.
    \item For $\tilde{y} \le y < y^*$: The constraint is binding, so $b_t^*=\bar{b}$.
    \item For $y \ge y^*$: The agent is retired, so labor supply is zero, $b_t^*=0$.
\end{itemize}

\item \textbf{Optimal Annuitization Time $\tau^*$:} This is an optimal stopping problem. The agent chooses to stop working and annuitize their wealth when the state process $y_t$ first reaches the optimal threshold $y^*$. This threshold is determined by the value-matching and smooth-pasting conditions, $V(y^*)=G(y^*)$ and $V'(y^*)=G'(y^*)$, ensuring an optimal transition. Therefore
    \[ \tau^* = \inf\{t \geq 0 : y_t \geq y^*\}. \]
\end{enumerate}
This completes the detailed derivation of the optimal policies based on the value function structure.
\end{proof}

\subsection{Proof of Proposition \ref{prop:optimal_threshold_y} (Optimal Retirement Wealth Threshold $y^*$)}\label{app:proof_y_star}
\begin{proof}
\noindent The optimal retirement wealth threshold $y^*$ is determined by ensuring a smooth and optimal transition from the working phase (governed by $V_{\bar{b}}(y)$) to the annuitized phase (governed by $G(y)$).

\paragraph{Boundary Conditions.}
For the transition to be optimal, the value function must be $C^1$ (continuously differentiable) across the boundary $y^*$. This imposes two conditions
\begin{enumerate}
    \item \textbf{Value Matching ($C^0$):} $V_{\bar{b}}(y^*) = G(y^*)$
    \item \textbf{Smooth Pasting ($C^1$):} $V'_{\bar{b}}(y^*) = G'(y^*)$
\end{enumerate}
Optimality and $C^2$ continuity also imply the super-contact condition $V''_{\bar{b}}(y^*) = G''(y^*)$.

\paragraph{The Maximized HJB Equation at the Boundary.}
The maximized HJB equation for the region $\tilde{y} \le y < y^*$ must hold in the limit as $y \to y^{*-}$. We evaluate the HJB at $y=y^*$ and substitute the boundary conditions. The HJB is
\begin{equation}\label{eq:HJB_at_boundary_proof_y}
\instdiscount V_{\bar{b}}(y^*) = u_1(\kappa^*, \bar{b}) + V'_{\bar{b}}(y^*) \left[ (r+\rho)y^* - \kappa^*(1+\rho y^*) + w \bar{b} \right] - \frac{1}{2}\theta^2 \frac{(V'_{\bar{b}}(y^*))^2}{V''_{\bar{b}}(y^*)},
\end{equation}
where $\theta = (\mu-r)/\sigma$. The annuitization function $G(y) = \frac{(ky)^{1-\gamma}}{\instdiscount(1-\gamma)}$ and its derivatives at $y^*$ are
\begin{align}
G(y^*) &= \frac{(ky^*)^{1-\gamma}}{\instdiscount(1-\gamma)}, \\
G'(y^*) &= \frac{k^{1-\gamma}}{\instdiscount} (y^*)^{-\gamma}, \\
G''(y^*) &= -\gamma \frac{k^{1-\gamma}}{\instdiscount} (y^*)^{-\gamma-1}.
\end{align}
The optimal consumption $\kappa^*$ at the boundary is determined by the FOC, using the smooth-pasting condition
\[
\frac{\partial u_1(\kappa^*, \bar{b})}{\partial \kappa} = V'_{\bar{b}}(y^*)(1+\rho y^*) = G'(y^*)(1+\rho y^*).
\]
By substituting $V_{\bar{b}}(y^*) = G(y^*)$, $V'_{\bar{b}}(y^*) = G'(y^*)$, $V''_{\bar{b}}(y^*) = G''(y^*)$, and the derived $\kappa^*(y^*)$ into the HJB equation \eqref{eq:HJB_at_boundary_proof_y}, the differential equation reduces to a single, non-linear algebraic equation. This equation provides one of the five constraints needed to solve for the system's unknowns.
\end{proof}

\subsection{Proof of Proposition \ref{prop:threshold_y_tilde} (Labor Constraint Threshold $\tilde{y}$)}\label{Proof_of_Derivation_of_tilde_y}

\begin{proof}
\noindent The threshold $\tilde{y}$ marks the boundary between the interior labor region ($V_{\mathrm{int}}$) and the corner labor region ($V_{\bar{b}}$). To ensure the overall value function is twice continuously differentiable ($C^2$), we must impose three continuity conditions at this boundary.

\paragraph{Boundary Conditions.}
Let $\tilde{\kappa}$ be the consumption-to-habit ratio at the boundary $\tilde{y}$. The conditions are
\begin{enumerate}
    \item \textbf{Value Matching ($C^0$):} $V_{\mathrm{int}}(\tilde{y}) = V_{\bar{b}}(\tilde{y})$
    \item \textbf{Smooth Pasting ($C^1$):} $V'_{\mathrm{int}}(\tilde{y}) = V'_{\bar{b}}(\tilde{y})$
    \item \textbf{Super-Contact ($C^2$):} $V''_{\mathrm{int}}(\tilde{y}) = V''_{\bar{b}}(\tilde{y})$
\end{enumerate}
These three equations form the first part of the complete system used to solve for the unknowns.

\begin{remark}[Characteristic Equations]
The characteristic equations arise from solving the second-order ordinary differential equations (ODEs) for the value function $V(y)$ in the continuation regions ($y < y^*$), as detailed in Appendix~\ref{Proof_of_Theorem_1:_Value_Function}. Because the HJB equation and optimal controls are different for the interior labor region ($V_{\mathrm{int}}$) and the corner labor region ($V_{\bar{b}}$), they each produce a \textit{distinct characteristic equation}. These equations are more complex than the standard \cite{merton1971optimum_na}'s form due to the state-dependent coefficients (e.g., the $(1+\rho y)$ term) in the HJB.

Nonetheless, each characteristic equation yields a set of distinct real roots (e.g., $m_1 > 0, m_2 < 0$ for $V_{\mathrm{int}}$ and $m'_1 > 0, m'_2 < 0$ for $V_{\bar{b}}$). The negative roots ($m_2, m'_2$) are particularly crucial for constructing the homogeneous solutions $A_2 y^{m_2}$ and $B_2 y^{m'_2}$ that satisfy the economic boundary conditions at $y=0$.
\end{remark}
\end{proof}

\subsection{System of Equations for Constants and Thresholds}

\noindent The five unknowns of the free boundary problem-- the integration constants $A_2, B_1, B_2$ and the thresholds $\tilde{y}, y^*$-- are determined by solving the following system of five non-linear algebraic equations derived from the continuity conditions.

\noindent
The general solutions for the value functions in the continuation regions are:
\begin{align*}
    V_{\mathrm{int}}(y) &= A_2 y^{m_2} + V_{p, \mathrm{int}}(y) \\
    V_{\bar{b}}(y) &= B_1 y^{m'_1} + B_2 y^{m'_2} + V_{p, \bar{b}}(y)
\end{align*}
where $m_2, m'_1, m'_2$ are the characteristic roots and $V_{p, \mathrm{int}}(y), V_{p, \bar{b}}(y)$ are the particular solutions for the respective ODEs. Unlike simpler models without habit formation, the complexity of the particular solutions $V_p$ prevents an explicit algebraic derivation of the constants $A_2, B_1, B_2$. The constants and thresholds are thus implicitly defined and solved for simultaneously.

\paragraph{Conditions at the Labor Constraint Boundary $\tilde{y}$.}
The $C^2$ continuity conditions at $y=\tilde{y}$ provide the first three equations
\begin{align}
    A_2 (\tilde{y})^{m_2} + V_{p, \mathrm{int}}(\tilde{y}) &= B_1 (\tilde{y})^{m'_1} + B_2 (\tilde{y})^{m'_2} + V_{p, \bar{b}}(\tilde{y}) \label{eq:sys1} \\
    m_2 A_2 (\tilde{y})^{m_2-1} + V'_{p, \mathrm{int}}(\tilde{y}) &= m'_1 B_1 (\tilde{y})^{m'_1-1} + m'_2 B_2 (\tilde{y})^{m'_2-1} + V'_{p, \bar{b}}(\tilde{y}) \label{eq:sys2} \\
    m_2(m_2-1) A_2 (\tilde{y})^{m_2-2} + V''_{p, \mathrm{int}}(\tilde{y}) &= m'_1(m'_1-1) B_1 (\tilde{y})^{m'_1-2} \nonumber \\ 
    & + m'_2(m'_2-1) B_2 (\tilde{y})^{m'_2-2} + V''_{p, \bar{b}}(\tilde{y}) \label{eq:sys3}
\end{align}

\paragraph{Conditions at the Annuitization Boundary $y^*$.}
The $C^1$ continuity conditions at $y=y^*$ provide the final two equations
\begin{align}
    B_1 (y^*)^{m'_1} + B_2 (y^*)^{m'_2} + V_{p, \bar{b}}(y^*) &= \frac{(ky^*)^{1-\gamma}}{\instdiscount (1-\gamma)} \label{eq:sys4} \\
    m'_1 B_1 (y^*)^{m'_1-1} + m'_2 B_2 (y^*)^{m'_2-1} + V'_{p, \bar{b}}(y^*) &= \frac{k^{1-\gamma}}{\instdiscount} (y^*)^{-\gamma} \label{eq:sys5}
\end{align}

\paragraph{Solution.}
The system of five non-linear equations \eqref{eq:sys1}-\eqref{eq:sys5} is solved numerically for the vector of unknowns $(\tilde{y}, y^*, A_2, B_1, B_2)$.

\end{document}